\documentclass[pra,twocolumn,showpacs,groupedaddress,notitlepage,superscriptaddress,floatfix,nofootinbib]{revtex4-2}

\usepackage{subfigure}
\usepackage{braket}
\usepackage{tikz}
\usepackage{amsmath}
\usepackage{listings}
\usepackage{algorithmicx}
\usepackage[ruled]{algorithm} 
\usepackage{algpseudocode}
\usepackage{amsfonts}
\usepackage{float,mathdots}
\usepackage{amssymb}
\usepackage{accents}
\usepackage{amsthm}
\usepackage{hyperref}
\usepackage{blkarray}
\usepackage{svg}

\newcommand\C{\mathcal{C}}
\newcommand\F{\mathbb{F}}

\newcommand\s{\mathbf{s}}
\newcommand\0{\mathbf{0}}

\renewcommand{\H}{\mathcal{H}}
\newcommand{\X}{\mathcal{X}}
\newcommand{\Z}{\mathcal{Z}}

\newtheorem{mytheorem}{Theorem}[section]
\newtheorem{mylemma}{Lemma}[section]
\newtheorem{mycorollary}{Corollary}[section]

\newtheorem{myclaim}{Claim}[section]
\newtheorem{mydefinition}{Definition}[section]

\DeclareMathOperator{\im}{im}

\DeclareMathOperator{\vect}{vec}

\renewcommand{\[}{\begin{equation}}
\renewcommand{\]}{\end{equation}}

\usepackage{bbm}
\usepackage{appendix}

\theoremstyle{definition}

\hypersetup{
    colorlinks=true,   
    linkcolor=cyan,    
    citecolor=magenta, 
    filecolor=magenta, 
    urlcolor=cyan,     
    runcolor=cyan
}
\usepackage[capitalise]{cleveref} 

\renewcommand\vec{\mathbf}

\newcommand{\wt}{\mathrm{wt}}
\newcommand{\I}{\mathbbm{1}}
\newcommand\aug{\fboxsep=-\fboxrule\!\!\!\fbox{\strut}\!\!\!}
\newcommand{\cyan}[1]{\textcolor{cyan}{#1}}

\newcommand{\DL}[1]{\cyan{[DL: #1]}}
\renewcommand{\DL}[1]{}

\begin{document}

\title{Bias-tailored single-shot quantum LDPC codes}

\author{Shixin Wu}
\affiliation{Department of Electrical \& Computer  Engineering,  University of Southern California, Los Angeles, California 90089, USA}
\affiliation{Center for Quantum Information Science \& Technology, University of
Southern California, Los Angeles, CA 90089, USA}
\author{Todd A. Brun}
\affiliation{Department of Electrical \& Computer  Engineering,  University of Southern California, Los Angeles, California 90089, USA}
\affiliation{Center for Quantum Information Science \& Technology, University of
Southern California, Los Angeles, CA 90089, USA}
\affiliation{Department of Physics \& Astronomy, University of Southern California,
Los Angeles, California 90089, USA}
\affiliation{Department of Computer Science, University of Southern California,
Los Angeles, California 90089, USA}
\author{Daniel A. Lidar}
\affiliation{Department of Electrical \& Computer  Engineering,  University of Southern California, Los Angeles, California 90089, USA}
\affiliation{Center for Quantum Information Science \& Technology, University of
Southern California, Los Angeles, CA 90089, USA}
\affiliation{Quantum Elements, Inc., Thousand Oaks, CA}
\affiliation{Department of Physics \& Astronomy, University of Southern California,
Los Angeles, California 90089, USA}
\affiliation{Department of Chemistry, University of Southern California,
Los Angeles, California 90089, USA}

\begin{abstract}
Quantum hardware rarely suffers equal amounts of bit-flip ($X$) and phase-flip ($Z$) errors; one type is often much more common than the other.  A code that is ``bias-tailored'' can exploit this imbalance, lowering the fault-tolerance overhead.  A complementary idea, called ``single-shot'' error correction, aims to recover from data errors and noisy measurements in a single round of stabilizer readout, avoiding slow repetition cycles.  In this work we combine these two ideas and build a hierarchy of new quantum codes.
The full construction starts from the syndrome-encoded hypergraph product code and then tailors it to the dominant error type.  The resulting code keeps the single-shot guarantee for every noise model while boosting the threshold whenever $X$ and $Z$ errors are asymmetric.
By removing carefully chosen blocks of stabilizers we obtain two trimmed variants.  The first, called the simplified code, cuts the physical-qubit count by $1/6$ and halves the number of stabilizer measurements, yet its minimum distance grows quadratically compared to the standard design and its biased noise threshold is unchanged.  The second, called the reduced code, achieves the same hardware savings but trades away single-shot protection for purely $X$ or purely $Z$ noise; instead it remains single-shot under balanced, or depolarizing, noise.  In settings where strongly biased noise is likely, either trimmed code offers a less resource-intensive alternative to the full construction.
As a concrete illustration, we lift the two-dimensional XZZX surface code to a three-dimensional cubic lattice and show that this ``3D XZZX'' code is an explicit member of the simplified family.
Taken together, these bias-tailored single-shot codes provide an adjustable set of code design alternatives, allowing tradeoffs between hardware overhead and noise types.
\end{abstract}

\maketitle

\section{Introduction}

The type of noise that afflicts quantum computers depends heavily on details such as architecture and qubit modality. In practice, physical qubits are more likely to be subject to biased noise where one type of error is stronger than the others. For example, in some superconducting qubit architectures $Z$ noise is stronger than $X$ noise by several orders of magnitude~\cite{aliferis2009fault,lescanne2020exponential, chamberland2022building}. 

\emph{Bias-tailoring quantum codes}.---
Most quantum error correction (QEC) codes~\cite{shor1995scheme,steane1996error,Calderbank:98,Lidar-Brun:book} were originally designed and evaluated without explicitly accounting for the possibility of noise bias. Although this is a historically common assumption, it leads to suboptimal performance compared to codes that exploit noise asymmetry~\cite{Aliferis:2008aa}. A \emph{bias-tailored} quantum code is designed under the assumption that noise is biased. It was shown in~\cite{puri2020bias} that, if implemented with bias-preserving two-qubit gates, bias-tailored codes can reduce the overhead required for fault-tolerant quantum computation. 

The XZZX surface code~\cite{bonilla2021xzzx} is a well-known example of a bias-tailored quantum code. It is constructed from a standard surface code via Hadamard rotation, which modifies local stabilizers. The XZZX surface code has a much higher threshold (50\%) under infinitely biased noise than the standard surface code (10.9\%)~\cite{bombin2012strong}; this high threshold is achieved because the XZZX code reduces to a set of decoupled repetition codes under infinitely biased noise.

Many efforts have been made to generalize the XZZX surface code~\cite{roffe2023bias,dua2022clifford, xu2023tailored, huang2022tailoring}. This work extends from the construction given by Roffe \emph{et al}.~\cite{roffe2023bias}, which generalized the bias-tailoring procedure by applying a Hadamard gate to the latter half of every stabilizer generator, a procedure they named ``Hadamard rotation", and which they applied to hypergraph product codes~\cite{Tillich13} and lifted-product codes~\cite{panteleev2021quantum} as examples. Both quantum codes are constructed from two copies of classical codes. Since these codes represent two major subfamilies of quantum low-density parity-check (qLDPC) codes~\cite{Breuckmann21}, the ones derived from applying Hadamard rotation to them are termed \emph{bias-tailored qLDPC codes}. Under infinitely biased noise, bias-tailored qLDPC codes reduce to multiple disjoint pieces of the classical codes that make up the qLDPC codes. In the XZZX case, these classical codes are repetition codes. 

\emph{Single-shot error correction}.---
QEC codes are vulnerable to noisy measurement; if the stabilizer generator measurements yield incorrect results, the decoder is less likely to successfully correct the errors that took place. In reality, measurement results are almost always unreliable. The standard approach to this problem is to repeat the same measurement multiple times and perform a majority vote on the results~\cite{shor1996fault}. However, this approach is susceptible to error accumulation during repeated measurements. As an alternative, Bombin proposed \emph{single-shot} fault-tolerant QEC, in which only a single round of noisy measurement is sufficient for a decoder to infer the required corrections~\cite{bombin2015single}. 

Studies in single-shot QEC revolve around two models of measurement errors: stochastic and adversarial noise. Under stochastic noise, a randomly selected constant fraction of qubits is corrupted by noise. Single-shot QEC in this setting requires QEC codes with additional structure. An example of such a structure is the high expansion of the associated factor graphs of these codes, first introduced by Leverrier \emph{et al}. in~\cite{leverrier2015quantum}, and later used to construct asymptotically good qLDPC codes~\cite{panteleev2022asymptotically,leverrier2022quantum,lin2022good}. The authors of~\cite{fawzi2020constant} realized single-shot QEC for quantum expander codes, which was later extended to quantum Tanner codes~\cite{gu2023single}. Under adversarial noise, single-shot QEC can be realized with redundant single-shot stabilizer generator measurements~\cite{ashikhmin2020quantum,delfosse2021beyond}. In ~\cite{campbell2019theory}, Campbell generalized the construction of quantum codes capable of single-shot QEC using homological products, first introduced by Freedman and Hastings~\cite{freedman2013quantumsystemsnonkhyperfinitecomplexes}.

\emph{This work}.---
Here, we combine bias-tailoring with single-shot QEC. Our construction ensures that single-shot QEC is possible under biased noise. We call the resultant code the \emph{bias-tailored syndrome-encoded hypergraph product} (BSH) code. Additionally, 
we propose two smaller, bias-tailored, versions of Campbell's homological products construction, which we call ``simplified'' and ``reduced''.
These codes achieve the same performance as Campbell's construction under infinitely biased or unbiased (depolarizing) noise, respectively. 
We also give a concrete example of the simplified version of Campbell's construction and show that it is equivalent to the three-dimensional XZZX code. 

This paper is structured as follows. In \cref{sec:2}, we introduce the preliminaries of coding theory and single-shot error correction. In \cref{sec:3}, we first introduce the homological product as a method to construct quantum codes, use bias-tailoring to create standard bias-tailored syndrome-encoded hypergraph product (BSH) codes, and show the single-shot properties due to bias-tailoring. Then, we give the construction of reduced bias-tailored syndrome-encoded hypergraph product (RSH) codes. In \cref{sec:4}, we give an explicit example of RSH codes, called the three-dimensional XZZX code, and discuss its properties. We conclude in \cref{sec:conc} and outline directions for future work.

\section{Preliminaries}

In this section we first briefly review the essentials of stabilizer codes and chain complexes. This material also serves to establish notation we use in the rest of this work. We then discuss biased noise models, and single-shot error correction.

\subsection{Classical and Stabilizer Codes}
\label{sec:2}

A classical linear code $\C$ is defined by a binary $m \times n$ parity-check matrix $H: \F_2^n \to \F_2^m$. $H$ has rank $r\le m$. The codespace of the code is the nullspace (kernel) of $H$. A code is characterized by three parameters: length $n$, number of logical bits $k$, and distance $d$. 
The number of logical bits $k$ is the dimension of the nullspace of $H$, which is $n-r$ by the rank-nullity theorem. The distance $d$ is the minimum Hamming weight [$\wt(\cdot)$] of a codeword, i.e., $d= \min_{x\in\mathrm{ker} H\setminus \{0\}}\wt(x)$. We overload these parameters with functions $n(\cdot)$, $k(\cdot)$, and $d(\cdot)$ that return the corresponding parameters. Their arguments can be either the code $\C$ or its defining parity-check matrix $H$. Additionally, we define the \emph{transpose code} $\C^T$ of a linear code $\C$ whose parity-check matrix is $H^T$. Similarly, it has parameters (functions) $n^T = m$, $k^T = m-r$ and $d^T = \min_{y\in\mathrm{ker} H^T\setminus \{0\}}\wt(y)$, where the $T$ superscript is used suggestively to denote the transpose code.\footnote{The transpose code differs from the dual code $\mathcal{C}^\perp$, which is generated by the rows of $H$ and has dimension $n-k$.}

The $1$-qubit Pauli group $\mathcal{G}\equiv \mathcal{G}^1$ is generated by $\{I,X,Y,Z\}$, where \[I=\begin{bmatrix} 1 & 0 \\ 0 & 1 \end{bmatrix},\ X=\begin{bmatrix} 0 & 1 \\ 1 & 0 \end{bmatrix},\ Y=\begin{bmatrix} 0 & -i \\ i & 0 \end{bmatrix},\ Z=\begin{bmatrix} 1 & 0 \\ 0 & -1 \end{bmatrix}.\]
Let $\mathcal{G}^n = (\mathcal{G}^1)^{\otimes n}$ denote the $n$-qubit Pauli group. 
The codespace of an $n$-qubit stabilizer code is the joint $+1$ eigenspace of a stabilizer group, where a stabilizer group $\mathcal{S}$ is an abelian subgroup of $\mathcal{G}^n$ that excludes $-I$ (we present the generalization to qudits in \cref{sec:qudit}). 
An error operator $E \in \mathcal{G}^n$ either commutes or anticommutes with each element $S_i \in \mathcal{S}$. Suppose the uncorrupted logical state is $\ket{\psi}$. If it commutes, $E\ket{\psi}=ES_i\ket{\psi}=S_iE\ket{\psi}$, which means $E\ket{\psi}$ is a $+1$ eigenstate of $S_i$. If this is true for all $S_i \in \mathcal{S}$, $E\ket{\psi}$ is in the codespace, and $E$ acts trivially on the logical state. Otherwise, there exists an $S_i$ such that $E\ket{\psi}=ES_i\ket{\psi}=-S_iE\ket{\psi}$, which means $E\ket{\psi}$ is a $-1$ eigenstate of $S_i$, and $E$ is a detectable error. The eigenvalues can be obtained by measuring the stabilizer generators, and the measurement results in $\{1,0\}^m$ are called the error syndrome, where $m$ is the number of stabilizer generators. In general, let $\sigma: \mathcal{G}^n \to \{0,1\}^m$ be the \emph{syndrome map}. If $s=\sigma(E)$, $s_i=1$ ($s_i=0$) if $E$ anticommutes (commutes) with the $i$'th stabilizer generator.

Note that $\sigma$ is a homomorphism: $\sigma(E_1 E_2)=\sigma(E_1)+\sigma(E_2)$ for $E_1,E_2 \in \mathcal{G}^n$. Suppose $E_2=E_1 S_i$ and $S_i \in \mathcal{S}$, then \[\sigma(E_2)=\sigma(E_1  S_i)=\sigma(E_1)+\sigma(S_i)=\sigma(E_1),\]
where the last equality follows because every stabilizer commutes with every stabilizer generator. We say that $E_1$ and $E_2$ are equivalent up to a stabilizer. As our following discussion focuses on the weight of Pauli operators, we define the reduced weights of equivalent Pauli operators up to stabilizers: 
\[
|E|^{\mathrm{red}}
  := \min\bigl\{\wt(F):\sigma(E)=\sigma(F), F \in \mathcal{G}^n\},
\label{eq:red}
\]
which is the smallest weight of the equivalent Pauli operators. 

The logical operators are the normalizers of $\mathcal{S}$ in $\mathcal{G}^n$. A logical operator commutes with all the stabilizers but is not in $\mathcal{S}$, so it acts non-trivially on a logical state. The distance of a stabilizer code is given by the minimum weight of its logical operators. 

Calderbank-Shor-Steane (CSS) codes~\cite{calderbank1996good,steane1996error,steane1997active} are well-known examples of stabilizer codes. Each stabilizer generator of a CSS code has either only $X$ or only $Z$ operators and is correspondingly called ``$X$-type'' or ``$Z$-type''. 

We use the binary vector representation for members of the Pauli group. For the generators of $\mathcal{G}$, 
\[\label{eq:rep} I \mapsto [0|0],\ X\mapsto[1|0],\ Y\mapsto [1|1],\ Z\mapsto[0|1].\]
Then, if $g \in \mathcal{G}^n$, $g \to [g_x|g_z]$, where $g_x, g_z \in \F_2^n$. 
If the $i$'th bit of $g_x$ is 1, $g$ applies an $X$ gate to the $i$'th qubit; if the $i$'th bit of $g_z$ is 1, $g$ applies a $Z$ gate to the $i$'th qubit; if the $i$'th bits of $g_x$ and $g_z$ are both 1, $g$ applies a $Y$ gate to the $i$'th qubit; if the $i$'th bit of $g_x$ and $g_z$ are both 0, $g$ acts trivially on the $i$'th qubit.
Hence, an $n$-qubit stabilizer code subject to $m$ stabilizer generators can be defined by a $m \times 2n$ binary matrix $H_Q=[H_X|H_Z]$, where $H_X, H_Z: \F_2^n \to \F_2^m$. An error $e_Q$  maps to $[e_X|e_Z] \in \F_2^{2n}$. When computing the syndrome we use the symplectic form, so the syndrome due to $e_Q$ is 
\[\label{eq:s_Q} s_Q = [H_X\mid H_Z]  \begin{bmatrix}e_Z \\ e_X\end{bmatrix} = H_Xe_Z + H_Ze_X\quad (\!\!\!\!\!\!\mod 2).\]
Because every stabilizer generator of CSS code contains $X$ or $Z$ operators, its binary matrix is a block matrix whose diagonal submatrices are zero matrices if we choose to write the $Z$-type stabilizer generators before the $X$-type:
\[
H_{\text{CSS}}=[H_X|H_Z]=\begin{bmatrix}
    0 & \aug & \mathcal{Z} \\
    \mathcal{X} & \aug & 0
\end{bmatrix}.
\]

The \emph{stabilizer weight} $w$ is defined as the maximum weight of any row and any column in the parity check matrices $\mathcal{X}$ and $\mathcal{Z}$. 
A family of codes is said to be \emph{low density parity check} (LDPC) if and only if $w = O(1)$ in the limit $n\to\infty$.\footnote{More precisely, when the two constituent classical codes are LDPC, every row (and column) of $H_X$ or $H_Z$ is bounded by a constant $w_0$. Concatenating the blocks into the full check matrix therefore increases those bounds by at most a factor of two, so all rows and columns have weight $\le 2w_0 = O(1)$. We will simply write $w=O(1)$ and regard this inessential factor $2$ duplication as implicit when referring to an LDPC family in the asymptotic limit $n\to\infty$.} 

The \emph{number of stabilizer checks} is the row count of $H_{\text{CSS}}$. This is the number of stabilizer measurements performed in every cycle of error correction. Note that this number is, in general, larger than the number of generators $r=n-k$ of a stabilizer code, which is the minimum number of measurements needed to detect all the errors the code can detect.
This redundancy is important for fault-tolerance to measurement noise in, e.g., single-shot protocols (discussed below), which need extra parity information to diagnose which measurement outcomes are faulty.

To check if the stabilizers commute as required for code validity, it is sufficient to check if $\mathcal{X}\mathcal{Z}^T=0$. The number of logical qubits can now be written as 
\[
\label{eq:k}
k = n - \mathrm{rank}\mathcal{Z} - \mathrm{rank}\mathcal{X} ,
\]
and the code distance as $d = \min\{d^X, d^Z\}$, where
\begin{equation}
 \label{eq:css_distance_defs} 
\begin{aligned}
d^X &= \min\{|v| \colon v \in \ker \mathcal{X} \setminus \mathrm{im}(\mathcal{Z}^T)\}, \\
d^Z &= \min\{|v| \colon v \in \ker \mathcal{Z} \setminus \mathrm{im}(\mathcal{X}^T)\} ,
\end{aligned}
\end{equation}
where $|v| = \wt(v)$ denotes the Hamming weight of the binary vector $v$, and $\ker(A)$ and $\mathrm{im}(A)$ denote the kernel and image, respectively, of the linear map $A$.

A popular method to construct a CSS code from two arbitrary linear codes without the orthogonality requirement of their parity-check matrices is to use the \emph{hypergraph product} proposed in~\cite{Tillich13}. 
The CSS code resulting from a hypergraph product is called a hypergraph product (HGP) code. 

The hypergraph product of two classical LDPC codes 
$[n_i,k_i,d_i]$ for $i=1,2$ 
yields an $[[n',k',d']]$ qLDPC CSS code whose parameters are
\[
  n'  =  n_1 n_2  +  m_1 m_2,\ \
  k'  =  k_1 k_2  +  k_1^{ \perp} k_2^{ \perp},\ \ d' \ge  \min\{d_1,d_2\},
\]
where $m_i$ is the number of parity checks in the $i$'th classical code and
$k_i^{\perp} = \dim\ker H_i^{\top}$ is the dimension of the dual
code.  
If each classical parity-check matrix $H_i$ has full row rank
($m_i = n_i - k_i$), then $k_i^{\perp}=0$ and the logical‐qubit count simplifies to $k' = k_1 k_2$.
In many standard constructions one also has $d' = \min\{d_1,d_2\}$.

We can use this formalism to characterize noise bias. On one end of the spectrum are the cases of pure $X$ or $Z$ noise. 
If there are only $Z$ errors, \cref{eq:s_Q} reduces to $s_Q=H_Xe_Z$, which is independent of $H_Z$. In this case, it suffices to require properties such as single-shot QEC and a pseudothreshold from either $H_X$ or $H_Z$ separately instead of both simultaneously. 
On the other end of the spectrum is depolarizing noise, where $X$ and $Z$ happen with the same probability, and we must consider $H_X$ and $H_Z$ simultaneously. In practice, we expect the noise bias to fall between these two extremes.

\subsection{Chain complexes and error correction codes}
\label{sec:3}

We review the construction of chain complexes and error correction codes, following \cite{freedman2013quantumsystemsnonkhyperfinitecomplexes,10.1145/2591796.2591870,hastings2016quantum,Breuckmann21}.

A length-$n$ \emph{chain complex} $C$ is a sequence of $n$ Abelian groups (``cells'') $\{C_i\}_{i\in[1,n]}$ 
and \emph{boundary operators} $\partial_i: C_{i} \to C_{i-1}$ that satisfy $\partial_i\partial_{i+1}=0$ (the boundary of a boundary is zero):
\[
  0 \xrightarrow{ \partial_{n+1}=0 } C_n \xrightarrow{ \partial_n }
  C_{n-1} \xrightarrow{ \partial_{n-1} } \cdots
  \xrightarrow{ \partial_1 } C_0 \xrightarrow{ \partial_{0}=0 } 0 ,
\]
We call $Z_i:=\ker \partial_i \subseteq C_i$ the \emph{$i$'th cycle} of $C$ and $B_i:=\im \partial_{i+1} \subseteq C_i$ the \emph{$i$'th boundary} of $C$. Note that $\partial_i\partial_{i+1}=0$ implies $B_i \subseteq Z_i$, so every boundary is a cycle, but not every cycle is a boundary.\footnote{To see this, assume $x\in B_i$; then by definition $x$ is a boundary and $x=\partial_{i+1}(y)$ for some $y\in C_{i+1}$. Then $\partial_i(x) = \partial_i\partial_{i+1}(y)=0$, so $x\in\ker\partial_i = Z_i$. It follows that $B_i\subseteq Z_i$.} The quotient group $\H_i:=Z_i/B_i$ is referred to as the \emph{i'th homology group}.

These concepts all have duals. A length-$n$ \emph{cochain} complex is the dual sequence
\[
  0 \xleftarrow{ \partial^{ n+1}=0 } C_n \xleftarrow{ \partial^{ n} }
  C_{n-1} \xleftarrow{ \partial^{ n-1} } \cdots
  \xleftarrow{ \partial^{1} } C_0 \xleftarrow{ \partial^{ 0}=0 } 0 ,
\]
with $\partial^{i}\partial^{i-1}=0$ and \emph{coboundary operators} $\partial^i:C_{i-1} \to C_{i}$ which are the dual operators of $\partial_i$.  We call $Z^i:=\ker \partial^{i+1} \subseteq C_i$ the \emph{i'th cocycle} of $C$ and $B^i:=\im \partial^{i} \subseteq C_i$ the \emph{i'th coboundary} of $C$. The \emph{i'th cohomology group} $\H^i$ is defined as $Z^i/B^i$.

Chain complexes can be used to define classical and quantum codes. A linear code defined by its parity-check matrix $H$ can be represented by a length-1 chain complex 
\[
X:=X_1 \xrightarrow{\partial_1[X]=H} X_0,
\]
where $X_1$ is identified as bits and $X_0$ as syndromes. Its first homology group $\H_1=\ker(H)/\mathbf{0}=\ker{H}$ is the space of codewords, and its zeroth homology group $\H_0=\F_2^m/\im(H)=\ker(H^T)$ is the space of the codewords of its transpose code. Then, $k(H)=\dim(\H_1)$ and $k(H^T)=\dim(\H_0)$. Formally, $\dim(\H_i)$ of chain complex $X$ is defined to be the $i$'th \emph{Betti number} $b_i(X)$, which is the number of $i$ dimensional holes in the space defined by $X$. For example, $b_1(X)=k(H)$. 

To connect with CSS codes, we encode the $Z$-type stabilizers in $\partial_2$ and the $X$-type stabilizers in $\partial_1$. We can represent $\partial_2$ and $\partial_1$ as matrices; $\partial_2=\Z^T$ has $n$ rows and has one column per $Z$-type stabilizer, while $\partial_1=\X$ has $n$ columns and one row per $X$-type stabilizer. In the general qudit case (\cref{sec:qudit}) the entries of these matrices are over $\mathbb{F}_p$, with $p\ge 2$ ($p=2$ is the qubit case). The entry in the $i$'th row and $j$'th column indicates which power of $Z$ appears in the $j$'th stabilizer; the entry in the $i$'th row and $j$'th column indicates which power of $X_j$ appears in the $i$'th stabilizer. The columns of $\partial_2$ are linearly independent, as are the rows of $\partial_1$, which ensures the linear independence of all stabilizers. 

The stabilizers must all commute with each other. Any pair of $Z$-type or $X$-type stabilizers trivially commute. The requirement that the $Z$-type stabilizers commute with the $X$-type stabilizers can be expressed as
\begin{equation} 
\label{eq:css_condition} 
\partial_1 \partial_2 = \X \Z^T = 0\qquad (\!\!\!\!\!\!\mod p).
\end{equation}
This requirement is equivalent to there being a chain complex
\[C_2\xrightarrow{\partial_2=\Z^T} C_1 \xrightarrow{\partial_1=\X} C_0,\]where $C_2, C_1, C_0$ are vector spaces over $\mathbb{F}_p$. The basis elements of $C_2$ are in one-to-one correspondence with $Z$-type stabilizers, those of $C_1$ with qudits, and those of $C_0$ with $X$-type stabilizers. Thus, we have $\dim(C_1) = n$.

The number of encoded qudits is given by the first Betti number, which is equal to 
\[
k = n - \dim(\partial_2) - \dim(\partial_1).
\]
Note the similarity to \cref{eq:k}.\footnote{We use the algebraic topology convention that $\dim(\partial_i) =  \text{rank}(\partial_i) =  \dim(\text{im}(\partial_i))$ for the linear map $\partial_i$.}

Furthermore, it is also possible to identify logical operators of a CSS code from a chain complex. The first homology group $\H_1=\ker(H_X)/\im(H_Z^T)$ is the space of logical $Z$ operators, and the first cohomology group $\H^1=\ker(H_Z)/\im(H_X^T)$ is the space of logical $X$ operators. 
Their dimensions $\dim(\H_1)$ and $\dim(\H^1)$ are the number of independent $Z$ and $X$ logical operators. Because we are working over a field, the universal coefficient theorem~\cite{Hatcher2002}
implies that $\H_1 \cong \H^1$ in a length-2 chain complex, so $\dim(\H_1)=\dim(\H^1)$, which means that a CSS code has the same number of logical $Z$ and $X$ operators. Therefore, if a length-2 chain complex represents a CSS code, its first Betti number is the number of logical $Z$ or $X$ operators.\footnote{Betti numbers are usually given with certain generating functions which are  Poincar\'{e} polynomials when their corresponding chain complexes are finite. 
An example is a 2-torus whose generating function is $(1+x)^2=1+2x+x^2$. Its first Betti number is the linear coefficient 2, which is equal to the number of logical qubits encoded by a toric code.}

A length-2 chain complex can also represent a classical linear code called the \emph{syndrome-encoded} (SE) code. An SE code is a classical code encoded by the parity-check matrix $H$, and the syndromes are encoded in the parity-check matrix $H_s$ with $H_sH=0$. On a length-2 chain complex, we identify the boundary operators with $\partial_2=H$ and $\partial_1=H_s$; the cells are identified such that $C_2$ is the bits, $C_1$ is the syndromes, and $C_0$ is the ``syndromes of syndromes''. If the syndromes are produced by $H$, we use ``syndromes of syndromes'' to correct them. We call the distance of the code defined by $H_s$ \emph{single-shot distance} $d_s$. Hence, an SE code has four parameters $[n,k,d,d_s]$. 

Two chain complexes can be combined via the \emph{tensor product}.\footnote{The usual tensor product of vector spaces encountered in quantum mechanics is the same basic algebraic construction used when forming the homological product of chain complexes. More specifically, one can take two chain complexes
$X \colon X_r \to X_{r-1} \to \dots \to X_0$ and $Y \colon Y_s \to Y_{s-1} \to \dots \to Y_0$, and form their (graded) tensor product
$X \otimes Y \colon \quad \bigoplus_{i+j=l} (X_i \otimes Y_j) \xrightarrow{\partial_l} \bigoplus_{i'+j'=l-1} (X_{i'} \otimes Y_{j'}) \xrightarrow{\partial_{l-1}} \dotsb$.
CSS codes come from a 2-step chain complex. The homological product of two CSS codes is just the $l=1$ part of the full graded tensor product of those two 2-step complexes.} 
As chain complexes represent codes, the tensor product of chain complexes offers us a structured method to combine codes. Formally, we define the codes obtained by a chain complex tensor product as \emph{homological product codes}. 

\begin{mydefinition}[Homological product codes~\cite{10.1145/2591796.2591870}]
    Let $\C_1$ and $\C_2$ be the linear codes defined by chain complexes $X$ and $Y$, respectively.
    The homological product code constructed from the homological product of $\C_1$ and $\C_2$, denoted by $\C_1 \otimes \C_2$, is the code represented by the chain complex as the result of the tensor product of $X$ and $Y$.
    We say $\C_1$ and $\C_2$ are the \emph{base codes} of $\C_1 \otimes \C_2$.
\end{mydefinition}
The homological product of two CSS codes $[[n_i, k_i, d_i]]$, 
$i=1, 2$, is an $[[n_1n_2, k_1k_2, d]]$  CSS code. The distances $d^X, d^Z$ of the product code [\cref{eq:css_distance_defs}] satisfy 
$\max\{d_1^\alpha, d_2^\alpha\} \le d^\alpha \le d_1^\alpha d_2^\alpha$, where $\alpha = X, Z$~\cite{10.1145/2591796.2591870}.

The stabilizer generators and the parameters of the homological product code can be easily obtained by applying theorems from algebraic topology. We define and discuss the chain complex tensor products and their related properties in the context of error correction codes in \cref{app:chain}. Here, we focus on the tensor product of two length-1 chain complexes and the tensor product of two length-2 chain complexes.\footnote{For this reason, the usual graded sign
$(-1)^{i}$ in the product boundary maps is invisible, so the block
expressions for $\partial_2[J]$, $\partial_1[J]$, and their
higher-degree analogues hold as written.  If we work over
$\mathbb{F}_p$ with $p$ odd, the second block in each product map
acquires a minus sign.}

Given length-1 chain complexes $X$ and $Y$, 
\[X:=X_1 \xrightarrow{\partial_1[X]} X_0,\ Y:=Y_1 \xrightarrow{\partial_1[Y]} Y_0.\]
Let $J=X \otimes Y$, then
\[\label{eq:graded-TP-J} 
\begin{aligned}
    J&=J_2 \xrightarrow{\partial_2[J]}J_1\xrightarrow{\partial_1[J]}J_0 \\
    &=X_1 \otimes Y_1 \xrightarrow{\partial_2[J]}\begin{bmatrix} X_0 \otimes Y_1 \\ X_1 \otimes Y_0 \end{bmatrix}\xrightarrow{\partial_1[J]}X_0 \otimes Y_0,
\end{aligned}\]
where
\[\partial_2[J]=\begin{bmatrix}
        \partial_1[X] \otimes \I \\ \I \otimes \partial_1[Y]
    \end{bmatrix},\ \partial_1[J]=\begin{bmatrix}
        \I \otimes \partial_1[Y]\ \partial_1[X] \otimes \I 
    \end{bmatrix}.\]
Given length-1 chain complexes $Z$ and $W$, $K=Z \otimes W$ is derived in the same fashion. Explicitly,
\[K:= K_2 \xrightarrow{\partial_2[K]} K_1 \xrightarrow{\partial_1[K]} K_0,\]
where
\[\partial_2[K]=\begin{bmatrix}
        \partial_1[Z] \otimes \I \\ \I \otimes \partial_1[W]
    \end{bmatrix},\ \partial_1[K]=\begin{bmatrix}
        \I \otimes \partial_1[W]\ \partial_1[Z] \otimes \I 
    \end{bmatrix}.\]
Now, let $Q=J \otimes K$, then
\[
\label{eq:21}
\begin{aligned}
    Q:=J_2 \otimes K_2 \xrightarrow{\partial_4[Q]} &J_1 \otimes K_2 \oplus J_2 \otimes K_1   \\
    \xrightarrow{\partial_3[Q]} &J_0 \otimes K_2 \oplus J_1 \otimes K_1 \oplus J_2 \otimes K_0   \\
    \xrightarrow{\partial_2[Q]} &J_0 \otimes K_1 \oplus J_1 \otimes K_0 \xrightarrow{\partial_1[Q]} J_0 \otimes K_0,
\end{aligned}
\]
where 
\[
\label{eq:l4bo}
\partial_4[Q]=\begin{bmatrix} \I\otimes\partial_2[K] \\ \partial_2[J] \otimes \I \end{bmatrix},\] 
\[\partial_3[Q]=\begin{bmatrix} 0 & \partial_1[J]\otimes \I\\ \partial_2[J]\otimes \I & \I \otimes \partial_2[K]\\ \I \otimes \partial_1[K] & 0\\\end{bmatrix}\]
\[\partial_2[Q]=\begin{bmatrix} \I \otimes \partial_2[K] & \partial_1[J]\otimes \I & 0 \\0 & \I\otimes \partial_1[K] & \partial_2[J]\otimes \I \end{bmatrix},\]
\[ \partial_1[Q]=[\I\otimes\partial_1[K]\ \partial_1[J]\otimes \I].\]
For our purposes, let $\{X,Y,Z,W\}$ represent classical linear codes, and $\{J,K\}$ represent the SE code. We call the resulting code $Q$ the \emph{syndrome-encoded hypergraph product} (SEHGP) code. $Q$ is a CSS version of the SE code. If we start from four classical linear codes, we obtain the SEHGP code after applying the tensor product of chain complexes twice. The flowchart of this construction is shown in \cref{fig:flowchart}. As the SEHGP code has encoded syndromes, it is also described by four parameters $[[n,k,d,d_s]]$.  

\begin{figure}
    \centering
    \includegraphics[width=0.7\linewidth]{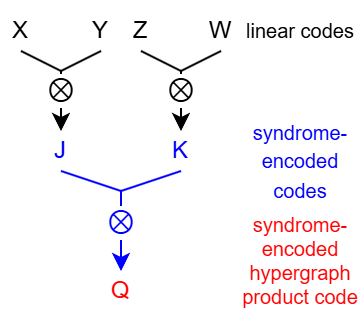}
    \caption{Flowchart of the construction of the SEHGP code, showing how four classical codes ($X,Y,Z,W$) are combined into two syndrome-encoded codes ($J,K$) and finally into the quantum code $Q$..}
    \label{fig:flowchart}
\end{figure}
We also provide two equations regarding the length and the number of logical qubits of the homological product code, which are discussed in detail in \cref{app:chain}.
For the code length:
\[
\label{eq:dim rule}
\dim(J_k)=\sum_{k=i+j} \dim(X_i) \dim(Y_j).
\]
For the number of logical qubits:
\[
\label{eq:betti rule}
b_k(J)=\sum_{k=i+j} b_i(X)b_j(Y),
\]
where $b_k(J)$ stands for the $k$'th Betti number of the $J$ chain complex.

Finally, we introduce a useful family of codes called Elias' \emph{direct product code}~\cite{Elias1954error}.
\begin{mydefinition}[Direct product code]
    Let $C_1 \boxtimes C_2$ be the direct product of two linear codes $\C_1$ and $\C_2$ of length (number of bits) $n_1$ and $n_2$, respectively. A codeword of the direct product code has length $n_1n_2$ and can be viewed as a matrix of dimension $n_1 \times n_2$ such that all its columns are codewords of $\C_1$, and all its rows are codewords of $\C_2$.
\end{mydefinition}
We claim that $\partial_2$ and $\partial_1^T$ of $J$ and $K$ are both direct product codes. We prove this claim for $\partial_2[J]$, and the proof for the rest is similar. More formally:
\begin{myclaim}
\label{claim:1}
The kernel of $\partial_2[J]$ is the direct-product code
$X\boxtimes Y$; equivalently,
$\ker\partial_2[J]=\{ \vect (M)\mid
  M\in X\boxtimes Y\}\subseteq \F_2^{n_1n_2}$.
The same statement holds with $X,Y$ replaced by $Z,W$ for $K$,
and with $\partial_2$ replaced by $\partial_1^T$.
\end{myclaim}

\begin{proof}
Let $\vect : \F_2^{n_1\times n_2}\to \F_2^{n_1n_2}$ be the
column-stacking map and recall the identity
\[
  (A\otimes B) \vect (C)= \vect\bigl(BCA^T\bigr).
\]

Direct-product $ \Rightarrow $ kernel:
Take any $M\in X\boxtimes Y$,
so each column of $M$ lies in $\ker\partial_1[X]$ and
each row lies in $\ker\partial_1[Y]$.
Hence
\[
  \partial_1[X] M = 0,
  \qquad
  M \partial_1[Y]^T=0 .
\]
Applying $\vect $ and the identity above,
\[
  (\partial_1[X] \otimes I_{n_2})\vect (M)=0,
  \quad
  (I_{n_1} \otimes \partial_1[Y])\vect (M)=0,
\]
which means $\vect (M)\in\ker\partial_2[J]$.

Kernel $ \Rightarrow $ direct-product:
Conversely, let $v\in\ker\partial_2[J]$ and set $M=\operatorname{unvec}(v)$.
The two block equations
\[
  (\partial_1[X] \otimes I_{n_2})v=0,
  \quad
  (I_{n_1} \otimes \partial_1[Y])v=0
\]
translate back to
$
  \partial_1[X] M = 0
$
and
$
  M \partial_1[Y]^T=0
$,
so the columns of $M$ lie in $X$ and its rows lie in $Y$.
Thus $M\in X\boxtimes Y$.

Combining the above two implications yields
$\ker\partial_2[J]=\{\vect (M):M\in X\boxtimes Y\}$.
The proofs for $\partial_1[J]^T$ and for the code $K$
follow by exchanging $X\leftrightarrow Y$ (or $Z\leftrightarrow W$)
and interchanging rows with columns.
\end{proof}

We give an example of this construction in \cref{app:example}.
As a result, the parameters of the code defined by $\partial_2[J]$ follow directly from the parameters of $X \boxtimes Y$, which are given by the following lemma.
\begin{mylemma}[\cite{macwilliams77}]
    \label{lem:direct-product}
    Let $\C_1$ be a $[n_1,k_1,d_1]$ code and $\C_2$ be a $[n_2,k_2,d_2]$ code. Then, $\C_1 \boxtimes C_2$ is a $[n_1n_2,k_1k_2,d_1d_2]$ code. 
\end{mylemma}
The proof of this lemma can be found in Appendix \ref{app:proof of the parameters of the Elias' product code}.

\section{Biased Noise Model}
\label{sec:4}

Throughout this work, we assume a noise model where each qubit is independently subject to a Pauli channel 
\begin{equation}
    \label{eq:channel}
    \mathcal{E}=(1-p)\rho+p_X X\rho X+p_Y Y\rho Y+p_Z Z\rho Z,
\end{equation}
where $\rho$ is the density matrix of the qubit, $X$, $Y$ and $Z$ are the Pauli matrices, and $p_X$, $p_Y$ and $p_Z$ are the probabilities of $X$, $Y$ and $Z$ errors such that $p_X+p_Y+p_Z=p$, $p$ is the total error probability. Under the independent noise model assumption, the $n$ qubit system is then subject to the noise channel $\mathcal{E}^{\otimes n}$. 

We define the \emph{noise bias} 
\[
\eta_i=\frac{p_i}{\sum_{j \neq i} p_j},\qquad i \in \{X,Y,Z\} .
\]
Depolarizing noise and infinitely biased noise are at opposite extremes of the noise bias range~\cite{tuckett2018ultrahigh, tuckett2019tailoring}. Under depolarizing noise $\eta_X=\eta_Y=\eta_Z=0.5$ and $p_X=p_Y=p_Z$, and each type of error is equally likely. Under infinitely $Z$-biased noise $\eta_Z=\infty$ and $\eta_X=\eta_Y=0$, and only $Z$ errors occur. 

The threshold of a quantum code is a maximum tolerable physical error rate below which the logical error rate can, in principle, be driven arbitrarily low by increasing the code size~\cite{knill1998resilient}. More formally, A threshold for a family of quantum error-correcting codes is a critical value $p_{\text{th}}$ such that, if the physical error rate $p<p_{\text{th}}$, then by increasing the code distance $d$ the logical error rate can be reduced arbitrarily. Conversely, if $p>p_{\text{th}}$, the logical error rate remains bounded away from $0$ even as $d\to\infty$. 
 
By mapping its code Hamiltonian onto a classical disordered eight-vertex Ising model, the toric code was shown by Bombin \emph{et al.} to have an error threshold of 16.4\% under depolarizing noise and of 10.9\% under infinitely biased noise~\cite{bombin2012strong} 
\DL{I couldn't find these threshold values in this reference}. 
By tailoring the topology of the toric code, the XZZX code maintains an error threshold of 
\DL{you wrote 16.4 but the paper \cite[p.3]{bonilla2021xzzx} says 18.7}
$18.7(1)\%$ under depolarizing noise while achieving 50\% under infinitely biased noise~\cite{bonilla2021xzzx}.

\subsection{Simplified Analysis: Identical Codes}
\label{sec:identical-codes}

We propose multiple families of codes in this paper and wish to compare their parameters. In this section, for simplicity, we assume all four of the base codes are the same $[n,k,d]$ code.  Moreover, this code is defined by an $n \times n$ parity-check matrix $H$ with $\ker H = \ker H^T$. The code and its transpose are thus identical, i.e., $n^T=n$, $k^T=k$, and $d^T=d$. In addition, the code's distance scales linearly with $n$. 

A simple example of this type of code is the ``closed-loop'' repetition code, i.e., a code whose
$n$ bits are arranged in a ring and the code requires that all bits be equal (i.e., the two codewords are all-$0$ and all-$1$).\footnote{A closed-loop repetition code is defined by a parity-check matrix in which one of these rows is always a linear combination of the others. For example, if $n=3$, the parity-check matrix is
$\Big(\begin{smallmatrix}
     1 & 1 & 0 \\
     0 & 1 & 1 \\
     1 & 0 & 1 \\
\end{smallmatrix}\Big)$.}
As a more general example, consider the general form of the parameters of the homological product code 
\[
[[n_1n_2+m_1m_2,k_1k_2+k_1^Tk_2^T,\min\{d_1,d_2,d_1^T,d_2^T\}]] .
\]
Under our simplifying assumption, this reduces to $[[2n^2,2k^2,d]]$. Let us compute the parameters of the SEHGP code (\cref{fig:flowchart}) under this simplification.

When we build the SEHGP code, we start with two 2-step complexes
$J$ and $K$ and tensor them once more (see \cref{fig:flowchart}), producing a four-step complex
\[
  Q_4 \xrightarrow{\partial_4} Q_3
      \xrightarrow{\partial_3} Q_2
      \xrightarrow{\partial_2} Q_1
      \xrightarrow{\partial_1} Q_0 ,
\]
where, reading off the different subspaces from \cref{eq:21}, we have:
\[
\begin{aligned}
Q_3 &= J_1 \otimes K_2 \oplus J_2 \otimes K_1\\
Q_2 &= J_0 \otimes K_2  \oplus J_1 \otimes K_1  \oplus J_2 \otimes K_0 \\
Q_1 &= J_0 \otimes K_1 \oplus J_1 \otimes K_0 .
\end{aligned}
\]
The physical qubits are identified with $Q_2$, so the
number of physical qubits is $\dim Q_{2}$.
$Q_3$ and $Q_1$ play the role of $Z$- and $X$-type parity checks, respectively. Hence, the sum of their dimensions is the total number of measurements (or checks). Finally, the number of encoded qubits equals the second Betti number $b_2(Q)=\ker\partial_2 \big/ \im\partial_3$. 

Next, let us calculate all these quantities explicitly.
Note that for every base code $X$ (and hence for $Y,Z,W$ as they are identical) we have a
length-1 chain complex
\[
  X_{1}\xrightarrow{\partial_{1}[X]=H}X_{0},
  \qquad \dim X_{1}=n, \dim X_{0}=n .
\]
The graded tensor-product rule \cref{eq:graded-TP-J} gives, for
$J=X\otimes Y\:\bigl($and identically for $K=Z\otimes W\bigr)$,
\[
\label{eq:37}
\begin{aligned}
J_{2}&=X_{1}\otimes Y_{1} &\Longrightarrow&\ \dim J_{2}=n^2,\\
J_{1}&=X_{0}\otimes Y_{1} \oplus X_{1}\otimes Y_{0}
     &\Longrightarrow&\ \dim J_{1}=2n^2,\\
J_{0}&=X_{0}\otimes Y_{0} &\Longrightarrow&\ \dim J_{0}=n^2.
\end{aligned}
\]
Hence
\[
\label{eq:38}
\begin{aligned}
  &\partial_{2}[J]\colon J_{2}\to J_{1}\quad
     \text{is an } 2n^2\times n^2\text{ matrix}\\
  &\partial_{1}[J]\colon J_{1}\to J_{0}\quad
     \text{is } n^2\times 2n^2.
\end{aligned}
\]
Exactly the same dimensions hold for $\partial_{2}[K]$ and
$\partial_{1}[K]$.

For any tensor product $J=X\otimes Y$ the Betti numbers are given by \cref{eq:betti rule}.
Because $\ker H=\ker H^T$, every base code has
$b_{1}(X)=b_{0}(X)=k$.
Applying \cref{eq:betti rule} to $J$ (and likewise $K$) yields
\[
\label{eq:39}
b_{2}(J)=k^{2},\quad 
 b_{1}(J)=2k^{2},\quad
b_{0}(J)=k^{2}.
\]

Using \cref{eq:dim rule,eq:betti rule} once more together with \cref{eq:37,eq:39},
\[
\begin{aligned}
\dim Q_{2}&= n^{2}n^{2}+2n^{2} 2n^{2}+n^{2}n^{2}=6n^{4},\\
b_{2}(Q)  &= k^{2}k^{2}+2k^{2} 2k^{2}+k^{2}k^{2}=6k^{4}.
\end{aligned}
\]
Thus, the
number of physical qubits is $\dim Q_{2}=6n^{4}$,
and the number of logical qubits is
$b_{2}(Q) = 6k^{4}$. 

With 
$\dim J_{1}=\dim K_{1}=2n^{2}$
and 
$\dim J_{0}=\dim J_{2}=\dim K_{2}=n^{2}$,
we have
\[
\label{eq:42}
\begin{aligned}
\dim Q_{3}
 &= \dim\bigl(J_{1} \otimes K_{2}\bigr) + \dim\bigl(J_{2} \otimes K_{1}\bigr)
  = 4n^{4}, \\
\dim Q_{1} &= \dim\bigl(J_0 \otimes K_1\bigr) + \dim\bigl(J_1 \otimes K_0\bigr) = 4n^{4}.
\end{aligned}
\]
Hence the SEHGP code family uses $8n^{4}$ stabilizer checks in total.

Finally, every non-trivial logical operator of the SEHGP code lives in $Q_2$. Each direct-sum component of $Q_2$ contains at least one tensor factor that is a logical operator of the identical base code, so the weight cannot drop below $d$.
Conversely, a weight-$d$ logical operator $L$ of the base code can be written as $L\otimes I \otimes I\otimes I \subset Q_{2}$. This shows
that the weight cannot increase above $d$.
Hence $d_{\text{SEHGP}} = d$.

Thus, under our simplifying assumption (all four base codes
identical) the SEHGP code family has parameters $[[6n^{4},6k^{4},d]]$.

\subsection{Single-shot Error Correction}
\label{sec:homo}

In this section, we briefly introduce single-shot error correction, following Campbell~\cite{campbell2019theory}. A decoder of a quantum code is an algorithm that outputs a recovery Pauli operator $E_r$ based on knowledge of the code and the measured syndrome vector $s$. Due to noisy measurement, $s=\sigma(E)+u$, where $\sigma(E)$ is the noiseless syndrome of the error $E$, and $u$ is the measurement error. A check set $\mathcal{M}\subset \mathcal{S}$ is a generator set of a stabilizer group $\mathcal{S}$. Good decoder design ensures that the residual error $E_r E$ due to noisy measurement is small. Formally, we define:
\begin{mydefinition}[Single-shot error correction,~\cite{campbell2019theory}]
Let $p$ and $q$ be integers and $f : \mathbb{Z} \to \mathbb{R}$ be some function with
$f(0) = 0$. We say a check set $\mathcal{M}$ is $(p, q, f)$ single-shot if
there exists a decoder such that for all syndrome errors $u$ and qubit errors $E$ such that $|u| < p$ and $f(2|u|) + |E| < q$, the decoder takes syndrome $s = \sigma(E) +u$ and outputs recovery operation $e$ such that $\min|E_r E| \leq f(2|u|)$, where the minimization is over multiplication by all elements of the code stabilizer, and $|P|$ denotes the weight of the Pauli operator $P$.
\end{mydefinition}

More rigorously, good single-shot properties require considering an infinite family of check sets.
\begin{mydefinition}[Good single-shot families,~\cite{campbell2019theory}]
    Consider an infinite check family $\mathcal{M}_n$ of $n$-qubit stabilizer codes. We say the family is a good single-shot family if each $\mathcal{M}_n$ is $(p, q, f)$ single-shot where $p, q \in \Omega(n^b)$ with $b > 0$, and $f(x)$ is some polynomial that is monotonically increasing with $x$ and independent of $n$.
\end{mydefinition}
In plain words, $p$ and $q$ grow polynomially with the size of the codes, and $f(x)$ remains constant with the size of the codes to remain constrained. 
The definition of single-shot error correction is somewhat obscure due to the existential condition of a decoder and the fact that there are many possible decoders. Alternatively,~\cite{campbell2019theory} borrows the concept of \emph{soundness} introduced in~\cite{aharonov2015quantum, hastings2016quantum} from locally testable codes.

\begin{mydefinition}[Soundness~\cite{campbell2019theory}]
    Let $t$ be an integer and $f : \mathbb{Z} \to \mathbb{R}$ some increasing function with $f (0) = 0$. We say that a stabilizer code has $(t, f )$ soundness if, for all errors Pauli $E$ with $|\sigma (E)| \leq t$,  
    \[
    f(|\sigma (E)|) \geq |E|^{\mathrm{red}},
    \]
   where the reduced weight was defined in \cref{eq:red}.
\end{mydefinition}

Similar to single-shot error correction, we can also define good soundness.
\begin{mydefinition}[Good soundness~\cite{campbell2019theory}]
    Consider an infinite family of stabilizer codes. We say that the family has good soundness if each code in it has $(t, f )$ soundness, where the following holds:
    \begin{enumerate}
        \item $t$ grows with the length $n$ of the code: $t = \Omega(n^b)$ with $b > 0$; 
        \item $f$ is monotonically increasing and independent of $n$.
    \end{enumerate}
\end{mydefinition}

Intuitively, a code has good soundness if small qubit errors produce small measurement syndromes. For example, Kitaev's 2D surface code~\cite{kitaev1997quantum} has bad soundness. 
To see this, recall that the 2D surface code is defined on a square lattice.
Consider two $Z$ stabilizer measurements that are distant on the square lattice, and both measurements yield $-1$. Suppose the cause is that
all qubits between these two stabilizers suffer $X$ errors. Hence, in this code, a low-weight syndrome can be produced by very high-weight qubit errors. 

One of the main results of \cite{campbell2019theory} is the connection between soundness and the success of single-shot error correction.
\begin{mytheorem}[Good soundness means good single-shot error correction~\cite{campbell2019theory}]
    \label{thm:good-soundness}
    Consider a quantum error correcting code with parameters $[[n, k, d, d_{s}]]$ and a check set that is $(t, f)$-sound. It is also $(p, q, f)$ single-shot where $p=\min\{d_{s}, t\}/2$ and $q=d/2$. 
\end{mytheorem}
Here $d_s$ denotes the single-shot distance of the code, defined in \cref{sec:3}. This result is important because it indicates that good soundness implies good single-shot error correction ability, which allows us to bypass the requirement for an explicit decoder when proving good single-shot ability. 

\subsection{Soundness from code homological product}
Besides creating new codes, the homological product of two codes can also give us the desired property of good soundness. 
In order to properly discuss soundness in the context of chain complexes, we use the following definition:

\begin{mydefinition}[Soundness of maps~\cite{campbell2019theory}]
Let $t$ be an integer and $f : \mathbb{Z} \to \mathbb{R}$ some increasing function with $f (0) = 0$. We say that a map $\partial$ has $(t, f )$ soundness if, for all binary vectors $E \in \F^n$ with $|\partial (E)| \leq t$, 
\[
f (|\partial (E)|) \geq \min\{\partial(E): \partial(E)=\partial(F), F \in \F^n\}.
\]
\end{mydefinition}
A linear code represented by a length-1 chain complex $C_1 \xrightarrow{\partial_1} C_0$ is $(t,f)$ sound if $\partial_1$ is $(t,f)$ sound because its syndrome map $\sigma$ is $\partial_1$. Hence, if $\partial_1$ has good soundness, the linear code it represents also has good soundness. 
Suppose a CSS code is represented by a length-2 chain complex $C_2 \xrightarrow{\partial_2} C_1 \xrightarrow{\partial_1} C_0$, 
where $\partial_2^T$ is $(t_2,f_2)$ sound, and $\partial_1$ is $(t_1,f_1)$ sound. Let $t=\min\{t_1,t_2\}$ and $f(x)=\min\{f_1(x),f_2(x)\}$ for all $x$. Then, the CSS code is $(t,f)$ sound because its syndrome map $\sigma$ is $\partial_1\oplus\partial_2^T$. If both $\partial_1$ and $\partial_2^T$ have good soundness, the CSS code has good soundness. 

We are now ready to state the first important lemma that connects soundness with the homological product of two codes. 
\begin{mylemma}[First soundness lemma~\cite{campbell2019theory}]
    \label{lemma: first soundness}
    Given a length-2 chain complex 
    \[J:= J_2\xrightarrow{\partial_2[J]}J_1 \xrightarrow{\partial_1[J]}J_0,\]
    where $J$ represents the homological product code $\C_1 \otimes \C_2$.
    The maps $\partial_2[J]$ and $\partial_1^T[J]$ are $(t,f)$ sound, where $t=\min\{d(\C_1),d(\C_2)\}$ and $f(x)=x^2/4$.
\end{mylemma}
We say that the soundness of $\partial_2[J]$ and $\partial_1^T[J]$ are ``for free'' because there is no requirement on the soundness of $\C_1$ and $\C_2$.

\DL{I inserted this Corollary here since this result is referenced later but wasn't formally derived.}

\begin{mycorollary}[Soundness of the SEHGP family]
\label{cor:SEHGP-soundness}
The composite syndrome map
$\partial_{2}\oplus\partial_{3}^{T}$ of every SEHGP code is
$(t,f)$-sound with
\[
t=\min\{d(X),d(Y),d(Z),d(W)\},\quad f(x)=x^{2}/4.
\]
Hence, the SEHGP family enjoys good soundness under depolarizing noise.
Because each summand $\partial_{2}$ and $\partial_{3}^{T}$ is itself
$(t,f)$-sound, the same bound holds when the channel produces only
$X$ errors or only $Z$ errors.
\end{mycorollary}

\begin{proof}
For the SEHGP complex we have two syndrome maps,
$\partial_{2}:Q_{2}\to Q_{1}$ and
$\partial_{3}^{T}:Q_{2}\to Q_{3}$.
Lemma~\ref{lemma: first soundness} applied to the two constituent
hypergraph products $J$ and $K$ tells us that each map is
$(t,f)$ sound with the same parameters
$t=\min\{d(X),d(Y),d(Z),d(W)\}$ and $f(x)=x^{2}/4$.

Write an error operator as $E=E_{X}\oplus E_{Z}$, where
$E_{X}$ is the part detected by $\partial_{2}$ and
$E_{Z}$ the part detected by $\partial_{3}^{T}$.  Its full syndrome is
\[
\begin{aligned}
\sigma(E)
  &= \bigl(\partial_{2}E_{X}\bigr)\oplus
     \bigl(\partial_{3}^{T}E_{Z}\bigr).
\end{aligned}
\]
Because weight is additive over a direct sum we have
$|\sigma(E)|=|\partial_{2}E_{X}|+|\partial_{3}^{T}E_{Z}|$.
If the depolarizing channel produces a syndrome of weight at most $t$,
then each component weight is also at most $t$.
Soundness of the individual maps gives pre-images
$R_{X},R_{Z}\in Q_{2}$ with
\[
\begin{aligned}
|R_{X}| &\le f(|\partial_{2}E_{X}|),\\
|R_{Z}| &\le f(|\partial_{3}^{T}E_{Z}|).
\end{aligned}
\]
Setting $R=R_{X}+R_{Z}$ recovers the full error and satisfies
\[
\begin{aligned}
|R|
  &\le |R_{X}|+|R_{Z}|
   \le f(|\partial_{2}E_{X}|)+f(|\partial_{3}^{T}E_{Z}|)\\
  &\le f\bigl(|\partial_{2}E_{X}|+|\partial_{3}^{T}E_{Z}|\bigr)
   = f\bigl(|\sigma(E)|\bigr),
\end{aligned}
\]
because $f(x)=x^{2}/4$ is convex and monotonically increasing.
Hence, the direct sum
$\partial_{2}\oplus\partial_{3}^{T}$ is $(t,f)$-sound.

When the channel produces only $X$ errors, the $\partial_{3}^{T}$ branch
never activates, so the syndrome reduces to $\partial_{2}E_{X}$ and the same
$(t,f)$ bound follows from the individual soundness of $\partial_{2}$.
The $Z$-only case is identical with the roles of the two maps swapped.
\end{proof}

\subsection{Bias-tailoring and commutation-preserving Hadamard rotation}

\DL{I tried to compress this section while preserving the content. See if you agree. The original text appears below, commented out.}

Hadamard rotation turns the ordinary surface code into the XZZX code by
interchanging some stabilizer supports between $X$ and $Z$.  We now
extend this idea to a broader class of codes via
\emph{bias-tailoring}.  Our starting point is the HGP code, of which the surface code is a special case.  Roffe
\emph{et al.} performed such a swap in \cite{roffe2023bias}, but without
an explicit algebraic justification.  We give that justification and
formulate a \emph{commutation-preserving Hadamard rotation} (CPHR) rule
that generalizes to SEHGP codes.

Write the $X$- and $Z$-type stabilizer blocks of an HGP code as
\[
H_{HGP}
=
\bigl[ H_X  \bigm|  H_Z \bigr]
=
\begin{bmatrix}
  \mathcal X_{1,1} & \mathcal X_{1,2} & \aug &
  \mathcal Z_{1,1} & \mathcal Z_{1,2} \\
  \mathcal X_{2,1} & \mathcal X_{2,2} & \aug &
  \mathcal Z_{2,1} & \mathcal Z_{2,2}
\end{bmatrix},
\]
with
\[
\mathcal X_{1,1}=0, 
\mathcal X_{1,2}=0, 
\mathcal Z_{1,1}=I_{n_1}\otimes H_2, 
\mathcal Z_{1,2}=H_1^T\otimes I_{m_2},
\]
\[
\mathcal X_{2,1}=H_1\otimes I_{n_2}, 
\mathcal X_{2,2}=I_{m_1}\otimes H_2^T, 
\mathcal Z_{2,1}=0, 
\mathcal Z_{2,2}=0.
\]

Throughout, we assume that the constituent classical codes satisfy the
orthogonality condition
$H_1 H_2^T=0$, which guarantees the original symplectic
commutation relation
$\mathcal X \mathcal Z^T=0$.

Let $R_k(H_{HGP})$ denote block row $k$.  Using
$R_2\bigl(H_{HGP}\bigr)$ and $R_1\bigl(H_{HGP}\bigr)$ we find
\[
\begin{aligned}
  R_2(H_{HGP}) R_1^T(H_{HGP})
  &= \mathcal X_{2,1}\mathcal Z_{1,1}^T
   + \mathcal X_{2,2}\mathcal Z_{1,2}^T \\
  &= H_1\otimes H_2^T
   + H_1 \otimes H_2^T = 0,
\end{aligned}
\]
so blocks commute as required.

We now define CPHR as exchanging one block column while leaving the block-row products above unchanged. More specifically, we define two types:

\emph{Type I CPHR (T1).}  Swap $\mathcal X_{2,1}\leftrightarrow
\mathcal Z_{1,1}$.  The new check matrix is
\[
H_{HGP}'
=
\begin{bmatrix}
  I_{n_1} \otimes H_2 & 0        & \aug & 0                  &
  H_1^T \otimes I_{m_2} \\
  0                     & I_{m_1} \otimes H_2^T &
  \aug &
  H_1 \otimes I_{n_2} & 0
\end{bmatrix},
\]
with updated logical blocks
\[
L_{HGP}'
=
\begin{bmatrix}
  \mathcal L_{Z1} & 0         & \aug & 0         & \mathcal L_{Z2} \\
  0               & \mathcal L_{X2} & \aug & \mathcal L_{X1} & 0
\end{bmatrix}.
\]

\emph{Type II CPHR (T2).}  Swap $\mathcal X_{2,2}\leftrightarrow
\mathcal Z_{1,2}$.  Now the new check matrix is
\[
H_{HGP}''
=
\begin{bmatrix}
  0                       & H_1^T \otimes I_{m_2} &
  \aug &
  I_{n_1} \otimes H_2 & 0 \\
  H_1 \otimes I_{n_2}  & 0 & \aug &
  0                       & I_{m_1} \otimes H_2^T
\end{bmatrix},
\]
exactly matching the bias-tailored matrix of \cite{roffe2023bias}.  The
logical operators become
\[
\begin{aligned}
L_{HGP}''
  &=
  \begin{bmatrix}
    0         & \mathcal L_{Z2} & \aug & \mathcal L_{Z1} & 0 \\
    \mathcal L_{X1} & 0         & \aug & 0               & \mathcal L_{X2}
  \end{bmatrix}.
\end{aligned}
\]

If $H_1 H_2^T=0$, then swapping either block column as in T1 CHPR
or T2 CPHR preserves the symplectic condition
$\mathcal X \mathcal Z^T=0$.
(The proof is a direct substitution in the block-row product shown
above.)

In $H_{HGP}''$ the $X$-type blocks form
$n_2$ disjoint copies of $H_1$ and $m_2$ disjoint copies of
$H_1^T$ along the block anti-diagonal, while the $Z$-type
blocks form disjoint copies along the block diagonal.
For example, with
$H_1=\bigl[\begin{smallmatrix}1&1&0\\0&1&1\end{smallmatrix}\bigr]$
and $m_1=3$,
\[
H_1 \otimes I_{n_2}
=
\begin{bmatrix}
  I_{n_2} & I_{n_2} & 0\\
  0       & I_{n_2} & I_{n_2}
\end{bmatrix}, 
I_{m_1} \otimes H_1^T
=
\begin{bmatrix}
  H_1^T & 0           & 0 \\
  0               & H_1^T & 0 \\
  0               & 0           & H_1^T
\end{bmatrix}.
\]
Thus each diagonal (or anti-diagonal) block acts on an independent
subset of qubits.

If the classical codes
$\{H_1,H_1^T,H_2,H_2^T\}$ each has 50\% threshold
(e.g., repetition codes, whose transpose is itself), the bias-tailored HGP inherits that threshold under infinitely biased noise because it is the direct sum of those disjoint classical instances
\cite{tuckett2019tailoring,bonilla2021xzzx}.
For a surface-code instance ($H_1=H_2$ repetition), the XZZX code
contains $2(n_2+m_2)$ diagonal repetition chains, so its threshold is again 50\%.

\subsection{Bias-tailoring the SEHGP code}
\DL{Same: I tried to compress this section while preserving the content. See if you agree. The original text appears below, commented out.}

The \emph{bias-tailored SEHGP code}, which we call the BSH code for short, is obtained by applying \emph{both} CPHR swaps
(T1 and T2) to the SEHGP stabilizer matrix.  In binary form the
stabilizers are
\[
  H_{SH}=\bigl[ H_X  \bigm|  H_Z \bigr],
\]
with
\[
\begin{aligned}
H_X &=
\begin{bmatrix}
  \mathcal O \\ \partial_2
\end{bmatrix}
=
\begin{bmatrix}
  0 & 0 & 0 \\
  0 & 0 & 0 \\
  \I \otimes \partial_2[K] & \partial_1[J] \otimes \I & 0 \\
  0 & \I \otimes \partial_1[K] & \partial_2[J] \otimes \I
\end{bmatrix}, \label{eq:SHHX}
\end{aligned}
\]
and
\[
\begin{aligned}
H_Z &=
\begin{bmatrix}
  \partial_3^T \\ \mathcal O
\end{bmatrix}
=
\begin{bmatrix}
  0 & \partial_2^T[J] \otimes \I &
      \I \otimes \partial_1^T[K] \\
  \partial_1^T[J] \otimes \I &
      \I \otimes \partial_2^T[K] & 0 \\
  0 & 0 & 0 \\
  0 & 0 & 0
\end{bmatrix}. \label{eq:SHHZ}
\end{aligned}
\]

Let us verify the commutation relation (row 3 vs row 2):
\[
\begin{aligned}
&R_3(H_{SH}) R_2^T(H_{SH})
  = \X_{3,1}\Z_{2,1}^T
     +\X_{3,2}\Z_{2,2}^T \\
  &\quad = \I \otimes \partial_2[K] 
     \partial_1[J] \otimes \I
     +\partial_1[J] \otimes \I 
      \I \otimes \partial_2[K] \\
  &\quad = 2 \partial_1[J] \otimes \partial_2[K]=0.
\end{aligned}
\]
The last equality holds over $\F_2$.

Next, let us apply the T2 CPHR to blocks $\X_{3,2},\Z_{3,2}$ and
$\X_{2,2},\Z_{2,2}$:
\[
  \X_{3,2}' \gets \Z_{3,2}, 
  \Z_{3,2}' \gets \X_{3,2}, 
  \X_{2,2}' \gets \Z_{2,2}, 
  \Z_{2,2}' \gets \X_{2,2}.
\label{eq:t2assign}
\]

Now, let us verify the commutation relations again (row 4 vs row 1):
\[
\begin{aligned}
&R_4(H_{SH}) R_1^T(H_{SH})
 = \X_{4,2}\Z_{1,2}^T+\X_{4,3}\Z_{1,3}^T \\
 &\quad = \I \otimes \partial_1[K] 
    \partial_2[J] \otimes \I
  + \partial_2[J] \otimes \I 
    \I \otimes \partial_1[K]                         \\
 &\quad = 2 \partial_2[J] \otimes \partial_1[K]=0.
\end{aligned}
\]

Because the T2 CPHR swap removed the $X$-$Z$ balance in the second block
column, we now swap the first block column in the complementary block
rows:
\[
  \X_{4,2}' \gets \Z_{4,2}, 
  \Z_{4,2}' \gets \X_{4,2}, 
  \X_{1,2}' \gets \Z_{1,2}, 
  \Z_{1,2}' \gets \X_{1,2}.
\label{eq:t1assign}
\]
A direct multiplication shows
\[
\begin{aligned}
R_3(H_{SH}') R_2^T(H_{SH}')  &=0,\\
R_4(H_{SH}') R_1^T(H_{SH}')  &=0,
\end{aligned}
\]
so commutation is preserved, as required.
\Cref{eq:t2assign,eq:t1assign} together specify the stabilizers of the
standard BSH code.

With $H_X'$ and $H_Z'$ defined, we are ready to define the parity-check matrices $H_{SX}'$ and $H_{SZ}'$ of the $X$ and $Z$ syndromes. They should satisfy $H_{SX}'H_X^{\prime T}=0$ and $H_{SZ}'H_Z^{\prime T}=0$ so that every row of $H_{SX}'$ ($H_{SZ}'$) is orthogonal to every
column of $H_X'$ ($H_Z'$). Moreover, since the $X$ and $Z$ syndromes are concatenations of the syndromes by the disjoint submatrices within $H_Z'$ and $H_X'$, $H_{SX}'$ and $H_{SZ}'$ must also be disjoint block matrices, and $H_{SX}' H_X'=0$ and $H_{SZ}'H_Z'=0$ must hold block-wise. Otherwise, if two sub-syndromes are each corrupted in one place, this total syndrome may still be considered valid. Below is an example of a bad construction:
\[\begin{aligned}
          & H_{SX}'=H_{SZ}'=[\partial_4^T\ \partial_1]                                                               \\
        &= [\partial_2^T[J]\otimes \I\ \I\otimes\partial_2^T[K]\ \I \otimes \partial_1[K]\ \partial_1[J]\otimes \I]
    \end{aligned}\]

Instead, to block-wise map $H_Z'$ and $H_X'$ to 0, we adopt:
\[
\begin{aligned}
    H_{SX}'&=\begin{bmatrix}
        \partial_2^T[J]\otimes \I & \I\otimes\partial_2^T[K] & 0                        & 0                        \\
        0                         & 0                        & \I \otimes \partial_1[K] & 0                        \\
        0                         & 0                        & 0                        & \partial_1[J] \otimes \I
    \end{bmatrix}\\
H_{SZ}'&=\begin{bmatrix}
        0                         & 0                        & \I \otimes \partial_1[K] & \partial_1^T[J] \otimes \I \\
        0                         & \I\otimes\partial_2^T[K] & 0                        & 0                        \\
        \partial_2^T[J]\otimes \I & 0                        & 0                        & 0                        \\
    \end{bmatrix}.
\end{aligned}
\]

Now, we give the parameters of the BSH code under the simplified analysis assumption of four identical codes (recall \cref{sec:identical-codes}). Its length, number of logical qubits and distance are obviously the same as the SEHGP code, so we only discuss its single-shot distance, soundness and threshold. 

\paragraph*{Single-shot distance.}
We discuss the distance of the code defined by $H_{SX}'$. The distance of $H_{SZ}'$ follows by symmetry.
By construction, it consists of three disjoint codes:
\[[\partial_2^T[J]\otimes \I \ \I\otimes\partial_2^T[K]],\ \I \otimes \partial_1[K],\ \partial_1[J] \otimes \I.\]
Thus, its distance is the smallest distance of the three codes. The code given by $\I \otimes \partial_1[K]$ is multiple copies of $\partial_1[K]$, so 
\[
\begin{aligned}
    d(\I \otimes \partial_1[K])&=d(\partial_1[K])\\
    &=\min\{d(Z),d(W),d^T(Z),d^T(W)\}.
\end{aligned}
\]
Similarly,
\[
\begin{aligned}
d(\partial_1[J] \otimes \I)&=d(\partial_1[J])\\
&=\min\{d(X),d(Y),d^T(X),d^T(Y)\}.
\end{aligned}
\]
To consider the distance of $[\partial_2^T[J]\otimes \I \ \I\otimes\partial_2^T[K]]$,
let $\partial_2^T[K]$ and $\partial_2^T[J]$ define two length-1 chain complexes
\[L=L_1 \xrightarrow{\partial_2^T[J]} L_0,\ R=R_1 \xrightarrow{\partial_2^T[K]} R_0.\]
Let $M=L \otimes R$ such that \[M=M_2 \xrightarrow{\partial_2[M]}M_1 \xrightarrow{\partial_1[M]} M_0,\]
where $\partial_1[M]=[\partial_2^T[J]\otimes \I \ \I\otimes\partial_2^T[K]]$. According to Ref.~\cite{Tillich13}, its distance is 
\[\min\{d(\partial_2[J]), d(\partial_2[K]),d^T(\partial_2[J]), d^T(\partial_2[K])\},
\]
i.e., by 
\[
\begin{aligned}
\min\{&d(X),d(Y),d^T(X),d^T(Y),\\
& d(Z),d(W),d^T(Z),d^T(W)\}.
\end{aligned}
\]
Because all four base codes are identical in our simplified analysis, each classical distance appearing in the minimum equals $d$; hence the single-shot distance of the BSH code is also $d$.

\paragraph*{Soundness.}
Now, we analyze the soundness of the code defined by $H_X'$; the analysis of the code defined by $H_Z'$ is similar. First, note that $H_X'$ consists of three disjoint codes whose parity-check matrices are: 
\[
A:=\I\otimes \partial_2[K],\ B:=\begin{bmatrix} \I \otimes \partial_2^T[K] \\ \partial_2^T[J]\otimes \I \end{bmatrix},\ C:=\partial_2[J]\otimes \I,
\]
The soundness of the code defined by $H_X'$ is determined by the least sound code among the three. The soundness of $A$ and $C$ can be determined by the following lemma, inspired by Claim 3 of~\cite{campbell2019theory}.
\begin{mylemma}[Inheritance of soundness]
    \label{lem:inheritance}
    If $\partial$ is $(t,f)$ sound, $\partial\otimes I_{n \times n}$ and $I_{n \times n} \otimes \partial$ are $(t,f)$ sound.
\end{mylemma}

\begin{proof}
We prove the claim for $\partial\otimes I_{n \times n}$; the result for $I_{n \times n} \otimes \partial$ follows by symmetry.

Let $q\in\im(\partial\otimes I_{n})$ with $|q|\le t$.  Because
$q=(\partial\otimes I_{n})(E)$ for some $E$, we may decompose
\[
  q=\sum_{i=1}^{n} \alpha_{i}\otimes e_{i},
\]
where $\{e_{i}\}$ are the standard basis vectors of $\mathbb F_{2}^{n}$
and each $\alpha_{i}\in\im\partial$ satisfies $|\alpha_{i}|\le t$.
Choose an index
\[
  i_{0} \in \arg\max_{1\le i\le n}|\alpha_{i}|.
\]
Because $\partial$ is $(t,f)$ sound, there exists
$\gamma_{i_{0}}\in\ker\partial$ with
$
  |\gamma_{i_{0}}|\le f(|\alpha_{i_{0}}|).
$
Set $\gamma_{i}=0$ for all $i\ne i_{0}$ and define
\[
  r=\gamma_{i_{0}}\otimes e_{i_{0}}.
\]
Then
\[
\begin{aligned}
  (\partial\otimes I_{n})(r) &= \partial(\gamma_{i_{0}})\otimes e_{i_{0}}
                              = 0,\\
  |r| &= |\gamma_{i_{0}}|
        \le f(|\alpha_{i_{0}}|)
        \le f\!\bigl(|q|\bigr),
\end{aligned}
\]
so $r$ satisfies the required inequality
$f(|q|)\ge|r|$.
Therefore, $\partial\otimes I_{n}$ is $(t,f)$ sound.
\end{proof}

According to the first soundness lemma, let $\partial_2[K]$ be $(t_K,f)$ sound and $\partial_2[J]$ be $(t_J,f)$ with $f(x)=x^2/4$. Then, $\I\otimes \partial_2[K]$ and $\partial_2[J]\otimes \I$ are $(t_K,f)$ and $(t_J,f)$ sound respectively. Recall that $J=X \otimes Y$ and $K=Z \otimes W$, then $t_J=\min\{d(X),d(Y)\}$ and  $t_K=\min\{d(Z),d(W)\}$.

Now, we discuss the soundness of $B$, which is the submatrix in the middle of $H_X'$. We use the $L$, $R$, and $M$ chain complex setting from our discussion in the single-shot distance.
Note that $\partial_2[M]=\begin{bmatrix} \I \otimes \partial_2^T[K] \\ \partial_2^T[J]\otimes \I \end{bmatrix} = B$.

The first soundness lemma states that $\partial_2[M]$ has good soundness and is $(t,f)$ sound with $t = \min \{d(\partial_2^T[J]), d(\partial_2^T[K])\}$ and $f(x)=x^2/4$, where~\cite{Tillich13}
\[
\begin{aligned}
 d(\partial_2^T[J])&=\min\{d(X),d(Y),d^T(X),d^T(Y)\}\\
d(\partial_2^T[K])&=\min\{d(Z),d(W),d^T(Z),d^T(W)\},
\end{aligned}
\]
Therefore, the $t$ parameter of the soundness of $H_X'$ is equal to 
\[
\begin{aligned}
    \min\{&d(X),d(Y),d^T(X),d^T(Y),\\
    &d(Z),d(W),d^T(Z),d^T(W)\},
\end{aligned}
\]
which is $d$ under our simplified analysis assumptions. Therefore, $H_X'$ and $H_Z'$ are $(d,x^2/4)$ sound, i.e., have good soundness. 

\paragraph*{Threshold.}
Under pure $X$ or pure $Z$ noise, the BSH code decomposes into independent copies of smaller classical codes, each of which presumably has a higher threshold. Consequently, we expect the BSH family to inherit a higher threshold in the infinite-bias limit than the original, non-bias-tailored SEHGP code; numerical verification is left for future work.

\subsection{Simplified and Bias-Tailored Simplified Syndrome-Encoded Hypergraph Product Codes}
\label{sec:simplified}

In this section, we use our simplified construction to obtain the \emph{simplified} SH (SSH) code and the \emph{bias-tailored} SSH (BSSH) code.

The first soundness lemma shows that in the length-2 chain complex
\[J=J_2 \xrightarrow{\partial_2[J]} J_1 \xrightarrow{\partial_1[J]} J_0,\]
$\partial_2[J]$ and $\partial_1^T[J]$ are $(t,f)$ sound where $t=\min\{d(X),d(Y)\}$ and $f(x)=x^2/4$. 
If $d(X)$ and $d(Y)$ scale polynomially with the code length of $X$ and $Y$, respectively, we already have good soundness. To preserve these soundness properties,
we define two classical linear $[n,k,d]$ codes with $\partial_2[J]$ and $\partial_1^T[J]$ represented by two length-1 chain complexes. Although we could, we do not mix and match the boundary operators from $K$ in this context for simplicity. As a result, we omit the $[J]$ label on all the boundary operators. The two length-1 chain complexes are:
\[
F=F_1 \xrightarrow{\partial_2} F_0,\ G=G_1 \xrightarrow{\partial_1^T} G_0,
\]
where, using \cref{eq:37},
\[
\label{eq:81}
\dim F_{1}=n^{2} = \dim G_{0},\qquad  \dim F_{0}= 2n^{2} = \dim G_{1} .
\]

The homological product code $P=F\otimes G$ is the \emph{simple syndrome-encoded hypergraph product} (SSH) code. It gives a 2-step complex
\[
  P_{2}\xrightarrow{\partial_{2}[P]}P_{1}
        \xrightarrow{\partial_{1}[P]}P_{0},
\]
with graded dimensions
\[
\begin{aligned}
\dim P_{2} &= \dim(F_{1}\otimes G_{1})
           = n^{2} \cdot 2n^{2}=2n^{4},\\
\dim P_{1} &= \dim(F_{0}\otimes G_{1})
           + \dim(F_{1}\otimes G_{0})                          \\
           &= 2n^{2} \cdot 2n^{2}+n^{2} \cdot n^{2}=5n^{4},\\
\dim P_{0} &= \dim(F_{0}\otimes G_{0})
           =2n^{2} \cdot n^{2}=2n^{4}.
\end{aligned}
\]
This homological product has stabilizer generators
\[
\label{eq:82}
    \begin{aligned}
        H_{SSH}&= [H_X|H_Z] \\
        =        & \begin{bmatrix}
                       0 & 0 &\aug & \partial_2 \otimes \I & \I \otimes \partial_1^T \\ 
                       \I \otimes \partial_1^T & \partial_2 \otimes \I &\aug &0 &0  \end{bmatrix} .
    \end{aligned}
    \]

We also bias-tailor the SSH code by applying T2 CPHR to the SSH code (the choice of T2 type is arbitrary). Doing so, we obtain
\[
    \begin{aligned}
        H_{BSSH}&= [H_X|H_Z] \\
        =        & \begin{bmatrix}
                       0 & \I \otimes \partial_1^T &\aug & \partial_2 \otimes \I & 0 \\ 
                       \I \otimes \partial_1^T & 0 &\aug &0 & \partial_2 \otimes \I  \end{bmatrix},
    \end{aligned}\]
which defines the \emph{bias-tailored SSH} (BSSH) code.

Recall the boundary operators requirement that $\partial_1\partial_2=\partial_2^T\partial_1^T=0$.
Using this rule, we can easily derive the parity-check matrices of the $X$ and $Z$ syndromes of the BSSH code.
\[H_{SX}=\begin{bmatrix}
    0 & \I \otimes \partial_2^T \\
    \I \otimes \partial_2^T  & 0\\
\end{bmatrix},\ 
H_{SZ}=\begin{bmatrix}
    \partial_1 \otimes \I & 0 \\
    0 & \partial_1 \otimes \I  \\
\end{bmatrix}.
\]

Next, we give the parameters of the SSH code. The parameters of the BSSH code are identical. 

\paragraph*{Number of physical qubits.} Recall that the physical qubits occupy the middle degree of the length-2 chain complex. Thus, 
the number of physical qubits is $\dim P_{1} = 5n^4$.

\paragraph*{Number of logical qubits.} Recall that the codes defined by $\partial_1^T$ and $\partial_2$ are both direct product codes (Claim~\ref{claim:1}). 
By Lemma~\ref{lem:direct-product}, these codes each encode $k^2$ logical bits. Thus, their homological product code encodes $(k^2)^2+(k^2)^2=2k^4$ logical qubits.

\paragraph*{Number of check operators.}
The number of check operators is the sum of the top and bottom degrees of the length-2 chain complex (the number of $X$- and $Z$-checks), i.e., $\dim P_2+\dim P_0$, or the row count of \cref{eq:82}, i.e.,
$2\times 2n^{4} = 4n^{4}$.

\paragraph*{Code distance.}  
By Lemma~\ref{lem:direct-product}, $d(\partial_2)=d(X)d(Y)$ and $d(\partial_1^T)=d^T(X)d^T(Y)$. Thus, the homological product of $\partial_2$ and $\partial_1^T$ has distance $\min\{d(X)d(Y),d^T(X)d^T(Y)\}=d^2$.

\paragraph*{Single-shot distance.} 
We analyze the single-shot distance of $H_{SZ}$; that of $H_{SX}$ follows by symmetry. Notice that $H_{SZ}$ consists of two disjoint copies of $\partial_1 \otimes \I$, and $\partial_1 \otimes \I$ is essentially multiple copies of $\partial_1$. Then, $d(H_{SZ})=d(\partial_1)$. Recall that $\partial_1$ is also the $X$ stabilizer generator matrix if we identify $J$ as an HGP code. Then $d(H_{SZ})=\min\{d(X),d(Y)\}=d$.

\paragraph*{Soundness.}
Under a depolarizing channel, our present arguments do not prove that the
SSH or BSSH families are $(t,f)$-sound: the mixed blocks in their
syndrome maps prevent a direct application of
Lemma~\ref{lemm:partial}, which guarantees soundness only for the
direct-sum syndrome map $\partial_2 \oplus \partial_1^{T}$.
In the pure $Z$ (or pure $X$) setting the BSSH stabilizer matrix
decouples, because only its $H_X$ block acts on errors; that block is two disjoint copies of $\I \otimes \partial_1^{T}$.
Since $\partial_1^{T}$ is $(t,f)$ sound, Lemma~\ref{lem:inheritance}
implies that $\I \otimes \partial_1^{T}$ is $(t,f)$-sound with the same function $f(x)=x^{2}/4$.  Hence, the BSSH code retains good soundness
whenever exactly one Pauli error type is present.  The SSH code lacks this guarantee, whereas the reduced SH (RSH) code we discuss in \cref{sec:reduced} below regains good soundness even for depolarizing noise.

\paragraph*{Threshold.}
Under pure $X$ or $Z$ noise, since the BSSH code is bias-tailored and consists of multiple copies of smaller codes, we expect the BSSH code to have a higher threshold than the SSH code. 

Under depolarizing noise, the BSSH code reduces to the SSH code. Hence, we expect them to have the same threshold.

\paragraph*{Summary.}
Recall from \cref{sec:identical-codes} that the SEHGP code family has parameters $[[6n^{4},6k^{4},d]]$.

The SSH construction reduces the qubit count from $6n^{4}$ to $5n^{4}$ and the number of stabilizer measurements from $8n^{4}$ to $4n^{4}$. It lifts the distance from $d$ to $d^{2}$.  At the same time, the number of logical qubits drops by a factor of three from $6k^{4}$ in the SEHGP family to $2k^{4}$ in the SSH and BSSH families. Under pure $X$ or $Z$ noise, the SSH code and the BSSH code achieve the same soundness as the SEHGP code. The threshold of the BSSH code is expected to be higher than the SSH code.

\subsection{Reduced and Bias-Tailored Reduced Syndrome-Encoded Hypergraph Product Codes}
\label{sec:reduced}

Recall (Corollary~\ref{cor:SEHGP-soundness}) that the SEHGP code has good soundness under depolarizing noise as well as under pure $X$ or $Z$ noise, which is advantageous.  Meanwhile, if the original base codes of the SEHGP code use $n$ qubits, the latter uses $6n^4$ qubits, which is a relatively high overhead. In this section, we propose a reduced construction that uses fewer qubits. This is accomplished by sacrificing the good soundness property under pure $X$ or $Z$ noise, which is justified, as this type of noise is very rare in practice.

Recall that \cref{eq:SHHX,eq:SHHZ} give the $X$ and $Z$ stabilizer generators of the SEHGP code. In particular, the third row of \cref{eq:SHHX} commutes with the second row of \cref{eq:SHHZ}.
The fourth row of \cref{eq:SHHX} commutes with the second row of \cref{eq:SHHZ}.
Hence, we can use them to define the \emph{reduced syndrome-encoded hypergraph product} (RSH) code with a pair of $X$ and $Z$ stabilizers $(H_{RSH1}, H_{RSH2})$ by truncating their zero blocks. We call the result RSH1 and RSH2 codes. These codes have the following stabilizer generators: 
\[
\begin{aligned}
H_{RSH1}&=[H_{RX1}|H_{RZ1}]\\
&=
\begin{bmatrix}
        0                        & 0                       & \aug & \partial_1^T[J] \otimes \I & \I \otimes \partial_2^T[K] \\
        \I \otimes \partial_2[K] & \partial_1[J]\otimes \I & \aug & 0                          & 0
    \end{bmatrix},
\end{aligned}
\]
and 
\[
\begin{aligned}
H_{RSH2}&=[H_{RX2}|H_{RZ2}]\\
&=
\begin{bmatrix}
        0                       & 0                       & \aug & \partial_2^T[J]\otimes \I & \I \otimes \partial_1^T[K] \\
        \I\otimes \partial_1[K] & \partial_2[J]\otimes \I & \aug & 0                         & 0
    \end{bmatrix}
\end{aligned}
\]

We now argue that \emph{at least one} of the two reduced codes,
RSH1 or RSH2, inherits the $(t,f)$-soundness of the full BSH code under
a noise model in which both $X$ and $Z$ errors occur.
Intuitively, every syndrome error activates either the ``left'' or ``right'' checks, so one of the two reduced matrices is
guaranteed to supply the necessary redundancy.
The formal proof is deferred to \cref{app:rsh}.

The drawback of the reduced construction is that we can no longer read
the syndrome parity-check matrices directly from a chain complex.
Nonetheless, they are easy to obtain.  Define $H_{RS1}$ and $H_{RS2}$ by
\[
H_{RS1}H_{RSH1}=0,\qquad
H_{RS2}H_{RSH2}=0 ,
\]
equivalently
\[
\ker H_{RS1}=\im H_{RSH1},\qquad
\ker H_{RS2}=\im H_{RSH2}.
\]
Because the kernel of a matrix is the orthogonal complement of its row
space, it suffices to take row-space complements of
$\im H_{RSH1}$ and $\im H_{RSH2}$ to build
$H_{RS1}$ and $H_{RS2}$.

The RSH code can also be bias-tailored.  Applying a type-T2
CPHR to $H_{RSH1}$ gives the
bias-tailored matrix
\[
\begin{bmatrix}
0 & \I\otimes\partial_2^{T}[K] & \aug &
  \partial_1^{T}[J]\otimes\I & 0 \\
\I\otimes\partial_2[K] & 0 & \aug &
  0 & \partial_1[J]\otimes\I
\end{bmatrix}
\]
and defines the bias-tailored reduced code BRSH1.  An analogous
rotation on $H_{RSH2}$ yields BRSH2.  Since the rotation decouples the
left and right checks, we expect the bias-tailored versions
to achieve higher thresholds under strongly biased noise.

We can next give the parameters of the RSH code; the parameters of the BRSH code are identical. Then, we analyze the soundness and the expected threshold of the BRSH code under pure $X$ or $Z$ noise and under depolarizing noise. Note that the BRSH code reduces to the RSH code, so we do not repeat our discussion in that regard. Recall once more that the SEHGP code is a $[[6n^4,4k^4,d]]$ code.

\paragraph*{Number of physical qubits.}
The matrix $H_{RX}$ is obtained from $H_{X}$ by deleting the block
$\partial_{2}[J]\otimes\I$, which maps
$J_{2}\otimes K_{0}\to J_{1}\otimes K_{0}$.
That block contributes $n^{4}$ columns and $2n^{4}$ rows.
After its removal, the middle degree of the reduced complex uses $5n^{4}$ physical qubits, which is the number describing the RSH family. This is $1/6$ fewer than the $6n^{4}$ qubits required by the SEHGP family.
The same count applies to the $Z$ sector by symmetry.

\paragraph*{Number of stabilizer checks.}
Eliminating $\partial_{2}[J]\otimes\I$ also removes half the rows of
both $H_{X}$ and $H_{Z}$.  The reduced pair
$(H_{RX},H_{RZ})$ therefore provides $4n^{4}$ check operators instead
of the $8n^{4}$ checks in the full SEHGP code.

\paragraph*{Number of logical qubits.}
Set
$U_{1}\xrightarrow{\partial_{1}[K]}U_{0}$ and
$V_{1}\xrightarrow{\partial_{2}[J]}V_{0}$.
The RSH1 stabilizers correspond to the length-2 complex
\[
\begin{aligned}
U_{1}\otimes V_{1}
  &\xrightarrow{H_{RZ}^{T}}
    U_{0}\otimes V_{1} \oplus U_{1}\otimes V_{0}\\
  &\xrightarrow{H_{RX}}
    U_{0}\otimes V_{0}.
\end{aligned}
\]
Each classical factor is a direct-product code, so
$k(\partial_{1}[K])=k^{2}$ and $k(\partial_{2}[J])=k^{2}$.
Applying the Betti number rule once gives
\[
\begin{aligned}
k(H_{RX})=k(H_{RZ}^{T}) = k(\partial_{1}[K])^{2}+k(\partial_{2}[J])^{2}
= 2k^{2}.
\end{aligned}
\]
A second application to the two-step complex shows that RSH encodes
$4k^{4}$ logical qubits, which is $1/3$ fewer than the $6k^{4}$ logical
qubits of the SEHGP code family.

\paragraph*{Code Distance.}
For RSH the two classical factors are
$\partial_{1}[K]$ and $\partial_{2}^{T}[J]$,
each a direct-product code with distance $d^{2}$.
The remaining blocks
$\partial_{1}^{T}[K]$ and $\partial_{2}[J]$ have distance $d$.
Hence
\[
\begin{aligned}
d(\text{RSH})
  &= \min\bigl\{d(\partial_{1}[K]),d(\partial_{2}^{T}[J]),
                d(\partial_{1}^{T}[K]),d(\partial_{2}[J])\bigr\} \\[2pt]
  &= \min\{d^{2},d^{2},d,d\}=d .
\end{aligned}
\]

\paragraph*{Single-shot Distance.} Because we have not found a closed form of $H_{RS1}$ to $H_{RS2}$, 
we cannot give the single-shot distances for the RSH codes. 

\paragraph*{Soundness.}
Under $X$-only or $Z$-only noise, the RSH blocks are
still coupled, so no uniform bound is available. However, the bias-tailored
version BRSH decouples: for pure $Z$ noise its $H_{X}$ part splits into
two disjoint matrices
$\I\otimes\partial_{2}[K]$ and $\I\otimes\partial_{2}^{T}[K]$.
Lemma~\ref{lem:inheritance} upgrades the good soundness of
$\partial_{2}[K]$ and $\partial_{2}^{T}[K]$ to the same $(t,f)$ bound
for their tensor products with $\I$, so BRSH has good soundness for
$Z$-only noise. An analogous argument applies to pure $X$ noise.

For depolarizing noise, \cref{app:rsh} shows that the full RSH pair
inherits the $(t,f)$ soundness of SEHGP, and BRSH reduces to RSH under a depolarizing channel, so both families share the $(t,f)$ soundness guarantee.

\paragraph*{Threshold.}
Under biased noise, BRSH is expected to outperform the RSH code
because its stabilizers align with the dominant error type.
Under depolarizing noise BRSH and RSH coincide, so their thresholds should be
identical.

\paragraph*{Summary.}
RSH drops the physical qubits cost from $6n^{4}$ for SEHGP to
$5n^{4}$ (a $1/6$ saving) and stabilizer checks by $1/2$ from $8n^{4}$ to
$4n^{4}$. The logical‐qubit count drops by $2/3$ from $6k^{4}$ to
$4k^{4}$, while the distance stays at $d$.  Both RSH and BRSH retain the
$(t,f)$ soundness of SEHGP for depolarizing noise; BRSH additionally
gains good soundness and (consequently) a higher threshold in the
single-bias limit.

\subsection{Three-dimensional XZZX Code}
\label{sec:3D-ZXXZ}

We close with a concrete BSSH example, the three-dimensional XZZX code, obtained by lifting the two-dimensional XZZX surface code of \cite{bonilla2021xzzx} from the square lattice to a cubic lattice.

We start from the closed-loop repetition code of length $n$ whose
parity-check matrix is $H_{\text{rep}}$ and note that its transpose
coincides with $H_{\text{rep}}$.
Let us insert this code into the SSH construction, exactly as in
\cref{sec:simplified}. The resulting stabilizer pair
\[
\begin{aligned}
H_{X} &=
\begin{bmatrix}
0 &
I \otimes
\begin{bmatrix}
I_{n} \otimes H_{\text{rep}}^{T}\\
H_{\text{rep}}^{T} \otimes I_{n}
\end{bmatrix}\\[4pt]
I \otimes
\begin{bmatrix}
I_{n} \otimes H_{\text{rep}}^{T}\\
H_{\text{rep}}^{T} \otimes I_{n}
\end{bmatrix} &
0
\end{bmatrix},\\[10pt]
H_{Z} &=
\begin{bmatrix}
\begin{bmatrix}
H_{\text{rep}} \otimes I_{n}\\
I_{n} \otimes H_{\text{rep}}
\end{bmatrix} \otimes I &
0\\[4pt]
0 &
\begin{bmatrix}
H_{\text{rep}} \otimes I_{n}\\
I_{n} \otimes H_{\text{rep}}
\end{bmatrix} \otimes I
\end{bmatrix}
\end{aligned}
\]
acts on $5 n^{4}$ physical qubits and has distance $d^{2}=n^{2}$,
exactly as described in \cref{sec:simplified}, with $k=2$. Note that the XZZX surface code has $k=1$ because of its non-periodic boundary condition while our code has the periodic boundary condition.

Let us now label the three tensor factors by $x$, $y$, and $z$.
Faces whose normal is parallel to $z$ ($xy$ faces) host the first block
row of $H_{X}$ and $H_{Z}$, while faces whose normal is parallel to
$y$ ($xz$ faces) host the second block row.
Faces parallel to $yz$ carry no stabilizers, so the lattice is only
partially populated. This construction is visualized in \cref{fig:3dxzzx}.

\begin{figure}
  \centering
    \includegraphics[width=0.8\linewidth]{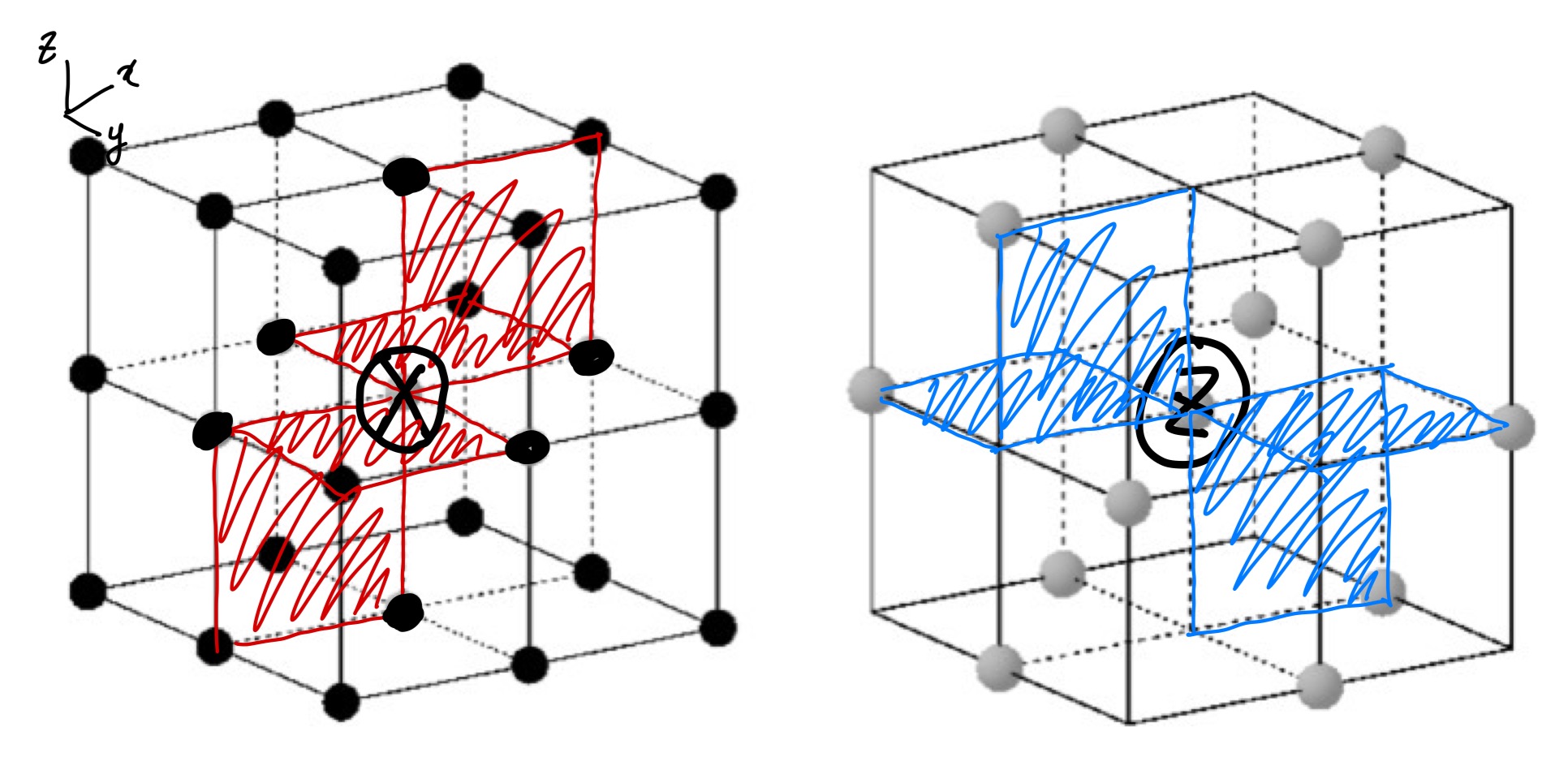}
  \caption{Layout of the three-dimensional XZZX code. Qubits live on the vertices of the cubic lattice. Red and blue faces mark stabilizers whose measurement flips under an $X$ or a $Z$ error, respectively.} 
  \label{fig:3dxzzx}
\end{figure}

\section{Discussion}
\label{sec:conc}

We have introduced a family of quantum error-correcting codes that balance three competing goals: low hardware overhead, tolerance to measurement noise in a single shot, and high performance when physical errors are strongly biased toward either $X$ or $Z$ errors. 
We briefly summarize the main ideas. 

Recall that a code is \emph{single-shot} if one round of noisy stabilizer measurements is already enough to correct both data and measurement errors.
\emph{Bias tailoring} means redesigning the code so that it works especially well when one Pauli error type (for example $Z$) is much more likely than the other. \emph{Hypergraph product} (HGP) codes build large quantum codes out of two small classical codes, while \emph{syndrome encoding} repeats the HGP step to protect the measurement outcomes themselves.

The code hierarchy we have introduced is the following:
\begin{itemize}
  \item \emph{SEHGP} (Syndrome-Encoded Hypergraph Product): the starting point.  It uses $6n^{4}$ physical qubits (starting from classical codes of distance $d$ over $n$ bits) and is single-shot for all realistic noise models, but it does not take advantage of noise bias. It yields $6k^{4}$ logical qubits and has distance $d$, the same as the underlying classical codes. It needs $8n^4$ stabilizer measurements.
  \item \emph{BSH}: a bias-tailored version of SEHGP that keeps all the guarantees of SEHGP, but improves the threshold when $X$ or $Z$ errors dominate.  It still needs $6n^{4}$ qubits.
  \item \emph{SSH} (Simplified SEHGP): drops the physical qubit count to $5n^{4}$ and halves the number of stabilizer measurements.  Its distance improves from $d$ to $d^{2}$, though the number of logical qubits falls from $6k^{4}$ to $2k^{4}$.
  \item \emph{BSSH}: the bias-tailored form of SSH; it is expected to surpass SSH when errors are highly asymmetric.
  \item \emph{RSH} (Reduced SEHGP): this is a different variation that starts from SEHGP, which again needs only $5n^{4}$ qubits and $4n^{4}$ checks.  It preserves the single-shot guarantee for depolarizing (unbiased) noise but not for purely $X$ or $Z$ noise.
  \item \emph{BRSH}: bias tailoring applied to RSH; it should outperform RSH under bias while matching it for unbiased noise.
\end{itemize}

By lifting the two-dimensional XZZX surface code of~\cite{bonilla2021xzzx} to a three-dimensional cubic lattice, we obtained a 3D XZZX code.  This code fits the SSH code with parameters $[[5n^{4},2,n^{2}]]$ and shows how our framework generalizes well-known XZZX surface codes.

Our results are purely analytic, so the next practical step is to build fast decoders and benchmark them in simulation. Existing single-shot strategies fall into two broad categories. Two-stage pipelines first repair the noisy syndrome and then correct data errors~\cite{brown2016fault,duivenvoorden2018renormalization,quintavalle2021single}. One-stage approaches act directly on the raw syndrome, usually with local update rules~\cite{kubica2018abcs,vasmer2021cellular,fawzi2018efficient,grospellier2018numerical}, although non-local schemes based on belief propagation also exist~\cite{breuckmann2021single,grospellier2021combining}. Quintavalle \emph{et al.}~\cite{quintavalle2021single} combined minimum-weight perfect matching (MWPM)~\cite{jack1965paths} and belief-propagation (BP) plus ordered-statistics decoding(BP+OSD)~\cite{panteleev2021degenerate,roffe2005decoding}, obtaining near-optimal thresholds for three-dimensional toric and surface codes under $Z$-only errors. Applying both their MWPM and BP+OSD stages to the SSH, BSSH, RSH, and BRSH families should reveal whether the same performance carries over. Because our constructions treat $X$ and $Z$ errors symmetrically, we expect similar gains under infinitely biased $X$ or $Z$ noise.

Besides BP+OSD, flip-style decoders are also worth exploring. The original flip algorithm was introduced for quantum expander codes~\cite{sipser1996expander}, and a probabilistic variant has recently reached competitive thresholds on 3D toric codes with $X$-only (face) errors~\cite{scruby2022local}. The probabilistic flip decoder typically needs $O(n^{2})$ time, slower than BP’s near-linear complexity, but its locality could simplify hardware implementation. A deeper challenge is to clarify how the flip mechanism, which relies on graph expansion, relates to the soundness condition used here to guarantee single-shot performance.

Campbell’s original framework~\cite{campbell2019theory} guarantees single-shot fault tolerance only against worst-case (``adversarial'') measurement noise. A general recipe for the more realistic setting of random (``stochastic'') noise is still missing. Quantum expander codes already achieve single-shot performance for stochastic errors without extra checks~\cite{fawzi2020constant,gu2023single}, but they rely on graph expansion rather than Campbell’s soundness. Bridging these two viewpoints remains open. If a unifying criterion can be identified, we could then apply the bias-tailoring techniques developed here and test whether the combined property survives.

The codes we have presented here offer different trade-offs: keep all guarantees at high cost (BSH), use half the checks for a quadratic distance boost (SSH/BSSH), or reduce both qubits and checks while keeping depolarizing soundness (RSH/BRSH).  Demonstrating fast decoders and experimentally testing thresholds will inform us which of these options is best suited for advancing quantum error correction subject to noise bias.

\acknowledgments
This material is based upon work supported by, or in part by, the U. S. Army Research Laboratory and the U. S. Army Research Office under contract/grant number W911NF2310255.





\appendix

\section{Qudit codes}
\label{sec:qudit}

Consider $p$-dimensional qudits. Let there be $n$ such qudits so that the total Hilbert space dimension is $p^n$. There are stabilizers of two types, called ``$X$-type'' and ``$Z$-type''. The $Z$ operator on the $p$-dimensional Hilbert space of a single qudit is 
\begin{equation} 
\label{eq:Z_operator} 
Z = 
\begin{pmatrix}
1 & & & \\
& \exp\left(\frac{2\pi i}{p}\right) & & \\
& & \ddots & \\
& & & \exp\left(\frac{2\pi i(p-1)}{p}\right)
\end{pmatrix},
\end{equation}
while the $X$ operator is 
\begin{equation} 
\label{eq:X_operator} 
X = 
\begin{pmatrix}
0 & 1 & 0 & \dots & 0 \\
0 & 0 & 1 & \dots & 0 \\
\vdots & \vdots & \vdots & \ddots & \vdots \\
0 & 0 & 0 & \dots & 1 \\
1 & 0 & 0 & \dots & 0 
\end{pmatrix}.
\end{equation}

Any power of $Z$ or $X$ with an exponent in the range $1, 2, \dots, d-1$ is a valid single-qudit stabilizer operator. The code subspace is the simultaneous $+1$ eigenspace of all the stabilizers, a subspace of the $p^n$-dimensional Hilbert space. The qubit case discussed in the main text corresponds to $p=2$.

\section{Chain Complex Tensor Product}
\label{app:chain}
In this appendix, we define the tensor product of two chain complexes, specifically when they are length-1. The extension to higher length can easily be given by induction.

First, let chain complexes $X$ and $Y$ represent two linear codes: 
\[X:=X_1 \xrightarrow{\partial_1[X]=H_1} X_0,\ Y:=Y_1 \xrightarrow{\partial_1[Y]=H_2} Y_0,\]
where $H_1$ and $H_2$ are $n_1 \times m_1$ and $n_2 \times m_2$ parity-check matrices that define $\C_1$ and $\C_2$ respectively.
Let $J=X\otimes Y$. Then, $J$ has cells $\{J_2,J_1,J_0\}$ where 
\[\label{eq:product rule}J_k=\bigoplus_{k=i+j}(X_i \otimes Y_j).\]
The boundary operators of $J$ are $\partial_2[J]: J_2 \to J_1$ and $\partial_1[J]: J_1 \to J_0$. Explicitly, as a chain complex and a cochain complex:
\[\label{eq:Zchain}J:=J_2 \xrightarrow{\partial_2[J]}J_1\xrightarrow{\partial_1[J]}J_0.\]
\[J^T:=J_2\xleftarrow{\partial_2^T[J]}J_1\xleftarrow{\partial_1^T[J]}J_0\]
For $x_i\in X_i$ and $y_i \in Y_i$ for $i\in \{1,2\}$, 
\[\begin{aligned}
    \partial_2[J](x_1 \otimes y_1)&=(\partial_1[X] x_1 \otimes y_1) \oplus (x_1 \otimes \partial_1[Y] y_1) \\
    &=(x_0 \otimes y_1 )\oplus (x_1 \otimes y_0).
\end{aligned}\]
More concisely, 
\[
\label{eq:partial2Z}
\partial_2[J]=\partial_1[X] \otimes \I \oplus \I \otimes \partial_1[Y].
\] 
Now, we define $\partial_1[J]$: 
\[
\begin{aligned}
    &\partial_1[J]\left((x_0 \otimes y_1) \oplus (x_1 \otimes y_0)\right) \\
    =&(x_0 \otimes \partial_1[Y]y_1)+(\partial_1[X]x_1 \otimes y_0) \\
    =&2(x_0\otimes y_0)=0.
\end{aligned}
\]
The last equality follows from the fact that $x_i$ and $y_i$ are binary vectors. This also proves that this definition is valid, as $\partial_1[J]\partial_2[J]=0$. The more concise way to write $\partial_1[J]$ is 
\[\label{eq:partial1Z} \partial_1[J]=((\I \otimes \partial_1[Y]) \oplus 0)+(0 \oplus (\partial_1[X] \otimes \I)\]. 

Recall the definition of the direct sum of vector spaces: for vector spaces $U$ and $W$, $V=U \oplus W$ if $V=U+W$ and $U \cap W=\{0\}$. In other words, there is a unique pair $(u,w) \in U \times W$ such that $v=u+w$ for every $v \in V$. As a result, the ``concatenated'' space of $U$ and $W$ is isomorphic to $U \oplus W$: 
\[\begin{bmatrix} U \\ W\end{bmatrix}=\left\{\begin{bmatrix} u \\ w\end{bmatrix}: u\in U,\ w \in W \right\}\cong U \oplus W.\]
With this isomorphism, \cref{eq:Zchain} becomes
\[X_1 \otimes Y_1 \xrightarrow{\partial_2[J]}\begin{bmatrix} X_0 \otimes Y_1 \\ X_1 \otimes Y_0 \end{bmatrix}\xrightarrow{\partial_1[J]}X_0 \otimes Y_0.\]
Correspondingly, from \cref{eq:partial2Z} and \cref{eq:partial1Z}, we rewrite $\partial_2[J]$ and $\partial_1[J]$ in their matrix forms:
\[
\partial_2[J]=\begin{bmatrix}
        \partial_1[X] \otimes \I \\ \I \otimes \partial_1[Y]
    \end{bmatrix},\ \partial_1[J]=\begin{bmatrix}
        \I \otimes \partial_1[Y]\ \partial_1[X] \otimes \I 
    \end{bmatrix}.
\]
Substituting $\partial_1[X]$ with $H_1$ and $\partial_1[Y]$ with $H_2^T$ in $\partial_2^T[J]$ and $\partial_1[J]$ gives us the stabilizer generators of the hypergraph product code defined in~\cite{Tillich13}. We remark that it requires us to exchange the left and right block columns of $\partial_2^T[J]$ and $\partial_1[J]$ to exactly obtain those defined in~\cite{Tillich13}. This is because we chose to use $X_0 \otimes Y_1 \oplus X_1 \otimes Y_0$ for our $J_2$. Since the order of the direct product does not matter, choosing $X_1 \otimes Y_0 \oplus X_0 \otimes Y_1$ gives us exactly those defined in~\cite{Tillich13}.
As a result, the homological product of two classical linear codes is equivalent to their hypergraph product.

Furthermore, we can obtain the length and the number of logical qubits with the chain complex tensor product formalism. By \cref{eq:product rule}, 
\[
\dim(J_k)=\sum_{k=i+j} \dim(X_i) \times \dim(Y_j) .
\]
The length of $X$ stabilizer generators $\partial_1[J]$ and length of $J$ stabilizer generators $\partial_2^T[J]$ are both given by 
\[\begin{aligned}
    \dim(J_2)&=\dim(X_1)\dim(Y_0)+\dim(X_0)\dim(Y_1) \\
    &=n_1m_2+m_1n_2,
\end{aligned}\]
K\"{u}nneth's theorem~\cite{hatcher2001algebraic} states that the Betti number generating function of $J$ is the product of those of $X$ and $Y$ \[p_{Z}=p_{X\times Y}=p_{X}p_{Y}.\] Hence, let $b_k(Z)$ be the $k$'th Betti number of $J$, then
\[
b_k(Z)=\sum_{k=i+j} b_i(X)b_j(Y).
\]
The number of qubits of the homological product code is 
\[
\begin{aligned}
    &\dim(\H_1(Z)) = b_1(Z)   \\
    &\quad =b_0(X)b_1(Y)+b_1(X)b_0(Y)  \\
    &\quad =\dim(\H_0(X))\dim(\H_1(Y))+\dim(\H_1(X))\dim(\H_0(Y))  \\
    &\quad =k_1^Tk_2+k_1k_2^T.
\end{aligned}
\]

The code distance correspond to the \emph{systole} (the shortest cycle in the manifold) of the resulting manifold due to homological product.

\section{Example of the Direct Product Code}
\label{app:example}
Suppose $\C_1$ is the $[3,1,3]$ repetition code and $\C_2$ is the
$[7,4,3]$ Hamming code, with parity-check matrices
\[
\begin{aligned}
H_1&=H_{\text{rep}}=\begin{bmatrix}1&1&0\\0&1&1\end{bmatrix}\\
H_2&=H_{\text{Ham}}=\begin{bmatrix}
0&0&0&1&1&1&1\\
0&1&1&0&0&1&1\\
1&0&1&0&1&0&1
\end{bmatrix}.
\end{aligned}
\]
Let \[H=\begin{bmatrix} H_1 \otimes I_{n_2} \\ I_{n_1} \otimes H_2 \end{bmatrix},\]
then 
\newcommand{\Intwo}{\begin{bmatrix} 1 & 0 & 0 \\ 0 & 1 & 0 \\ 0 & 0 & 1 \end{bmatrix}}
\newcommand{\zeronone}{\begin{bmatrix} 0 & 0 & 0 \\ 0 & 0 & 0 \\ 0 & 0 & 0 \end{bmatrix}}

\[U:=H_1 \otimes I_{7}=\begin{bmatrix}
        I_{7} & I_{7} & 0 \\ 0 & I_{7} & I_{7}
    \end{bmatrix}_{14 \times 21},\]
and
    \[L:= I_{3} \otimes H_2=\begin{bmatrix}
        H_2 & 0 & 0 \\ 0 & H_2 & 0 \\ 0 & 0 & H_2
    \end{bmatrix}_{9 \times 21}.\] 
From $L$, rows 1-7, 8-14, and 15-21 of the 21-bit vector must each satisfy the Hamming parity checks, while from $U$ the triples 
$$(1,8,15),(2,9,16),\dots,(7,14,21)$$ must satisfy the repetition-code checks.
As discussed in \cref{sec:homo}, bits 1-21 will be rearranged in the following form:
\[\begin{matrix}
1  &  2  &  3  &  4  &  5  &  6  &  7 \\
8  &  9  & 10  & 11  & 12  & 13  & 14 \\
15 & 16  & 17  & 18  & 19  & 20  & 21 \\
\end{matrix}\]
If the bits are arranged in such a form, we can now confirm that the columns are the codewords of the repetition code ($\C_1$) and the rows are codewords of the Hamming code ($\C_2$).

\section{Good soundness of the RSH code}
\label{app:rsh}

Let $r\in Q_{2}$ and
\[
\begin{aligned}
\begin{bmatrix}\partial_2\\\partial_3^{T}\end{bmatrix} r
   = 
  \begin{bmatrix}s_{1}\\s_{2}\end{bmatrix}
   =  s .
\end{aligned}
\]
Partition $r$ and the two syndrome halves as
\[
  r=\begin{bmatrix}r_a\\ r_b\\ r_c\end{bmatrix},
  \qquad
  s_{1}=\begin{bmatrix}s_{L1}\\ s_{R1}\end{bmatrix},
  s_{2}=\begin{bmatrix}s_{L2}\\ s_{R2}\end{bmatrix},
\]
with
\[
\begin{aligned}
s_{L1}&=(\I \otimes \partial_2[K]) r_a
         +(\partial_1[J] \otimes \I) r_b,\\
s_{R1}&=(\I \otimes \partial_1[K]) r_b
         +(\partial_2[J] \otimes \I) r_c,\\
s_{L2}&=(\partial_2^{T}[J] \otimes \I) r_b
         +(\I \otimes \partial_1^{T}[K]) r_c,\\
s_{R2}&=(\partial_1^{T}[J] \otimes \I) r_a
         +(\I \otimes \partial_2^{T}[K]) r_b.
\end{aligned}
\]
Whenever the measured syndromes are \emph{uncorrupted},
\[
  \partial_1 s_1 = 0,
  \qquad
  \partial_4^{T} s_2 = 0,
\]
which implies the consistency constraints
\[
\begin{aligned}
m_1&:=(\I \otimes \partial_1[K])s_{L1}
     =(\partial_1[J] \otimes \I)s_{R1},\\
m_2&:=(\partial_2^{T}[J] \otimes \I)s_{L2}
     =(\I \otimes \partial_2^{T}[K])s_{R2}.
\end{aligned}
\]
Substituting yields
\[
\begin{aligned}
  m_1&=(\partial_1[J] \otimes \partial_1[K]) r_b,\\
  m_2&=(\partial_2^{T}[J] \otimes \partial_2^{T}[K]) r_b,
\end{aligned}
\]
so the intermediate vector $r_b$ alone determines both consistency
checks.

\begin{mylemma}[cf.\ Lemma 7 of \cite{campbell2019theory}]
\label{lemm:partial}
There exists a choice of $r_b$ obeying
\[
\begin{aligned}
  m_1&=(\partial_1[J] \otimes \partial_1[K]) r_b,
  & |r_b|&\le |s_{L1}| |s_{R1}|,\\[4pt]
  m_2&=(\partial_2^{T}[J] \otimes \partial_2^{T}[K]) r_b,
  & |r_b|&\le |s_{L2}| |s_{R2}|,
\end{aligned}
\]
and such that the remainders
\[
  s_{L1}-(\partial_1[J] \otimes \I)r_b, 
  s_{R1}-(\I \otimes \partial_1[K])r_b,
\]
\[
  s_{L2}-(\partial_2^{T}[J] \otimes \I)r_b, 
  s_{R2}-(\I \otimes \partial_2^{T}[K])r_b
\]
decompose into elementary tensors  
$\sum_i \alpha_i\otimes a_i$ with
$\alpha_i\in\ker\partial_1[J]$ (or
$\ker\partial_2^{T}[J]$, respectively) satisfying
$
  \eta(\alpha_i), |\alpha_i|\le |s_{L1}|
$
(and analogously for the other three remainders).
\end{mylemma}
The proof is identical to that of Lemma 7 in
\cite{campbell2019theory}, replacing the two base maps there by
$\partial_1,\partial_2$ and their transposes.

A corollary immediately follows this lemma:
\begin{mycorollary}
\label{cor:mixed-bounds}
If
$
  |r_b|\le |s_{L1}| |s_{R1}|
$
''{and}
$
  |r_b|\le |s_{L2}| |s_{R2}|,
$
then at least one of the ``mixed'' bounds
\[
  |r_b|\le |s_{L1}| |s_{R2}|
  \quad\text{or}\quad
  |r_b|\le |s_{L2}| |s_{R1}|
\]
must also hold.
\end{mycorollary}

\begin{proof}
Multiplying the two given inequalities yields
\[
\begin{aligned}
  |r_b|^{2}
   \le 
  |s_{L1}| |s_{R1}| |s_{L2}| |s_{R2}|.
\end{aligned}
\]
Suppose, for contradiction, that both mixed bounds fail, i.e.,
$|r_b|>|s_{L1}| |s_{R2}|$ and
$|r_b|>|s_{L2}| |s_{R1}|$.
Multiplying these assumed violations gives
\[
\begin{aligned}
  |r_b|^{2}
  &> 
  |s_{L1}| |s_{R2}| |s_{L2}| |s_{R1}|
  \\
  &= 
  |s_{L1}| |s_{R1}| |s_{L2}| |s_{R2}|,
\end{aligned}
\]
contradicting the earlier upper bound.  Hence at least one mixed
inequality must be satisfied.
\end{proof}

\smallskip
\noindent

Assume, without loss of generality, that the bound
$|r_b|\le |s_{L1}| |s_{R2}|$ holds.  We show that the left-hand reduced code RSH1 inherits good
soundness; the argument for RSH2 is symmetric.  Restrict attention to
\[
  r'=\begin{bmatrix} r_a\\ r_b \end{bmatrix},
  \qquad
  s'=\begin{bmatrix} s_{L1}\\ s_{R2} \end{bmatrix}.
\]

Choose a vector $r_a$ of minimal weight that satisfies
\[
\begin{aligned}
(\I \otimes \partial_2[K]) r_a
     &= s_{L1}-(\partial_1[J] \otimes \I) r_b,\\
(\partial_1^{T}[J] \otimes \I) r_a
     &= s_{R2}-(\I \otimes \partial_2^{T}[K]) r_b.
\end{aligned}
\]
Because $\partial_2[K]$ and $\partial_1^{T}[J]$ are each $(t,f)$-sound with $f(x)=x^{2}/4$, the following lemma shows that their tensor products with the identity retain a scaled-up soundness.



Applying Lemma~\ref{lem:inheritance} and the ``small remainder''
properties of Lemma~\ref{lemm:partial} yields
\[
  |r_a| \le |s_{L1}| f(|s_{L1}|)
  \quad\text{and}\quad |r_a|\le |s_{R2}| f(|s_{R2}|).
\]
Hence, using $f(x)=x^{2}/4$ and the triangle inequality,
\[
\begin{aligned}
|r'|
  &= |r_a|+|r_b| \\
  &\le \frac14|s_{L1}|^{3} + |s_{L1}| |s_{R2}|
       + \frac14|s_{R2}|^{3} \\
  &\le \frac14\bigl(|s_{L1}|+|s_{R2}|\bigr)^{3}
    =  \frac14 |s'|^{3}.
\end{aligned}
\]

The inequality above shows that any error whose syndrome (in RSH1) has
weight at most $t$ is accompanied by a recovery of weight at most
$\tfrac14|s'|^{3}=g(|s'|)$, where $g(x)=x^{3}/4$.
Because $t=\Omega(n^{b})$, the pair
$\bigl(t,g\bigr)$ satisfies the criteria of good soundness.
Thus, the RSH1 code is $(t,g)$-sound; the same argument, exchanging
left and right blocks, proves soundness for RSH2 whenever the mixed
bound $|r_b|\le |s_{L2}| |s_{R1}|$ holds.

\section{Proof of the parameters of the Elias' product code}
\label{app:proof of the parameters of the Elias' product code}
In this section, we prove the parameters of Elias' product code. This proof is inspired by \cite{Varodayan2002InvestigationOT}.
\begin{mylemma}[\cite{macwilliams77}]
    Let $\C_1$ be a $[n_1,k_1,d_1]$ code and $\C_2$ be a $[n_2,k_2,d_2]$ code. Then, $\C_1 \boxtimes C_2$ is a $[n_1n_2,k_1k_2,d_1d_2]$ code. 
\end{mylemma}
\begin{proof}
    Let $C=C_1 \boxtimes C_2$ be an $[n, k, d]$ code. First, we prove that $C$ remains to be a linear code. Recall that the codewords of $C$ are matrices whose columns are the codewords of $C_1$ and whose rows are the codewords of $C_2$. The linear combination of the codewords of $C$ amounts to the linear combination of their columns and rows, which are the codewords of $C_1$ and $C_2$. Since $C_1$ and $C_2$ are linear codes, the linearly combined columns and rows remain the codewords of $C_1$ and $C_2$. Therefore, the linear combination of the codewords of $C$ is still a codeword of $C$, so $C$ is a linear code. 

    It is obvious that $n=n_1n_2$ from the definition of the codewords of $C$. Now, we prove $d=d_1d_2$. A non-zero column $\vec{a}$ of a codeword of $C$ has weight at least $d_1$, and a non-zero entry of a non-zero column is the component of a non-zero row $\vec{b}$, which is of weight at least $d_2$ because it is a non-zero codeword of $C_2$. Therefore, for every non-zero entry of a column, there are at least $d_2$ non-zero entries in its corresponding row, so the total weight of a codeword of $C$ is at least $d_1d_2$. Then, if $\vec{a} \in C_1$ has weight exactly $d_1$ and $\vec{{b}} \in C_2$ has weight exactly $d_2$, they form a valid codeword of weight $d_1d_2$ of $C$, which proves the distance of $C$ is $d_1d_2$.

    Finally, we prove $k=k_1k_2$. Let $G_1$ and $G_2$ be the generator matrices of $C_1$ and $C_2$ where $\vec{g}_i^1$ and $\vec{g}_j^2$ are the rows of $G_1$ and $G_2$. First, we prove that a codeword of $C$ is the outer product of $\vec{g}_i^1$ and $\vec{g}_j^2$: 
    \[
        M:=(\vec{g}_i^1)^T (\vec{g}_j^2) \in C
    \]
    Let
    \[
        \vec{g}_i^1=[a_1, \dots,a_{n_1}] \in C_1,\ \vec{g}_j^2=[b_1, \dots,b_{n_2}] \in C_2
    \]
    Then, a matrix component $M_{p,q}=a_pb_q$. Fix a column $q$,
    \[
        M_{:,q}=[a_1b_q, \dots, a_{n_1}b_q]^T=b_q[a_1, \dots,a_{n_1}] \in C_1,
    \]
    because $C_1$ is a linear code closed under scalar multiplication. Similarly, fix a row $p$,
    \[
        M_{p,:}=[a_pb_1, \dots, a_pb_{n_2}]=a_p[b_1, \dots,b_{n_2}] \in C_2,
    \]
    because $C_2$ is also a linear code. 

    Furthermore, we can write
    \[
        (\vec{g}_i^1)^T (\vec{g}_j^2)=G_1^TE_{ij}G_B,
    \]
    where $E_{ij}$ is a $k_1\times k_2$ matrix where its $(i,j)$th entry is 1 and the rest is 0. Then, we want to show the codewords are linearly independent. Suppose some linear combination of them equals zero:
    \[
        \sum_{i,j} \lambda_{ij} G_1^TE_{ij}G_B=G_1^T\left(\sum_{i,j}\lambda_{ij}E_{ij} \right)G_B=\vec{0}.
    \]
    Then, if we multiply both sides by $(G_1^T)^{-1}(\cdot)G_2^{-1}$, we get
    \[
        \sum_{i,j}\lambda_{ij}E_{ij}=\vec{0}
    \]
    Since $E_{ij}\neq \vec{0}$, then $\lambda_{ij}=0$ for all $i \in [1,k_1]$ and $j \in [1,k_2]$. Hence, we have proved that the codewords are linearly independent, and we are left to prove that the codewords span the codespace of $C$.

    Let  $K$ be some codeword in $C$ in its matrix form. The columns of $M$ are in the rowspace of $G_1$, and the rows of $M$ are in the rowspace of $G_2$. We can use either fact to proceed, and we use the last fact. This means that there is some $n_1 \times k_2$ matrix $D$ such that $K=DG_2$, which gives us $D=CG^{-1}_2$. This expression means that every column of $D$ is in the columnspace of $C$ which is the same of the rowspace of $G_1$. As a result, $D=G_1^TE$, where
    \[
        E=\sum_{i,j} \lambda_{ij}E_{ij}
    \]
    for some $\lambda_{ij} \in \{0,1\}$. This is possible because $E_{ij}$ is all zero but a 1 at position $(i,j)$. Combining these statements gives
    \[
        K=DG_1=G_1^TEG_2=G_1^T\left(\sum_{i,j} \lambda_{ij}E_{ij}\right)G_2,
    \]
    which expands to
    \[
        \sum_{i,j} \lambda_{ij}G_1^TE_{ij}G_2
    \]
    This means that $K$ is a linear combination of the codewords.

\end{proof}

\bibliography{bib}

\end{document}